%% file: main.tex
\documentclass[USenglish]{llncs}
\usepackage{cite}
\makeatletter
\RequirePackage[bookmarks,unicode,colorlinks=true]{hyperref}%
   \def\@citecolor{blue}%
   \def\@linkcolor{blue}%

\def\orcidID#1{\smash{\href{http://orcid.org/#1}{\protect\raisebox{-1.25pt}{\protect\includegraphics{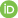}}}}}
\makeatother 

\usepackage[dvipsnames]{xcolor}
\usepackage[utf8]{inputenc}
\usepackage{amsmath,% amsthm,
  amssymb}
\usepackage{xcolor} 
\usepackage[inference]{semantic}
\usepackage{stmaryrd}
\usepackage{mathpartir}
\usepackage{booktabs}
\usepackage{siunitx}
\usepackage{stmaryrd}
\usepackage{thm-restate}
\usepackage{multicol}
\usepackage{color}
\usepackage{algorithm}
\usepackage{float}
\usepackage{multirow}
\usepackage{graphicx}
\usepackage{tikz}
\usepackage{mathpartir}
\usepackage{xspace}
\usepackage{paralist}
\usepackage[inline]{enumitem}
\usepackage{hyperref}

\usepackage[capitalise,nameinlink]{cleveref}
\usetikzlibrary{calc}
\usetikzlibrary{arrows,automata}
\usetikzlibrary{positioning}
\usepackage{ifthen}
\usepackage{xargs}
\usepackage{wrapfig}
\usepackage{pifont}
\usepackage{xcolor}
\usepackage{etoolbox}
\usepackage[inline]{enumitem}
\usepackage{adjustbox}
\usepackage{changepage}
\usepackage{graphicx,calc}
\usepackage[T1]{fontenc}
\input{defs}

\newif\iflong
\longtrue % uncomment to compile long version
\newif\ifdraft
\drafttrue % uncomment to compile with comments
\newcommand{\mytitle}{Rely-Guarantee Reasoning for Causally Consistent Shared Memory (Extended Version)}

\title{\mytitle\thanks{Lahav is supported by the Israel Science
Foundation (grants 1566/18 and 814/22) and by the European Research Council
(ERC) under the European Union’s Horizon 2020 research and innovation programme
(grant agreement no. 851811).
Dongol is supported by EPSRC grants EP/X015149/1, EP/V038915/1, EP/R025134/2, VeTSS, and ARC Discovery Grant DP190102142. Wehrheim is supported by the German Research Council DFG (project no. 467386514).}}

\author{Ori Lahav\inst{1}\orcidID{0000-0003-4305-6998} \and
Brijesh Dongol\inst{2}\orcidID{0000-0003-0446-3507} \and
Heike Wehrheim\inst{3}\orcidID{0000-0002-2385-7512}}

\institute{Tel Aviv University, Tel Aviv, Israel \and
University of Surrey, Guildford, UK \and  
University of Oldenburg, Oldenburg, Germany}

\begin{document}

 \maketitle

%\begin{center}
%\scalebox{1.4}{\textbf{Rely-Guarantee Reasoning for}} 
%\smallskip
%
%\smallskip
%\scalebox{1.4}{\textbf{Causally Consistent Shared Memory}}
%\end{center}
%\vspace{-5pt}

\begin{abstract}
  Rely-guarantee (RG) is a highly influential compositional proof
  technique for concurrent programs, which was originally developed
  assuming a sequentially consistent shared memory.  In this paper, we first
  generalize RG to make it parametric with respect to the underlying
  memory model by introducing an RG framework that is applicable to
  any model axiomatically characterized by Hoare triples.  Second, we
  instantiate this framework for reasoning about concurrent programs
  under \emph{causally consistent memory}, which is formulated using a
  recently proposed \emph{potential-based} operational semantics,
  thereby providing the first reasoning technique for such semantics.
  The proposed program logic, which we call \lora, employs a novel
  assertion language allowing one to specify ordered sequences of states that
  each thread may reach.  We employ \lora for multiple litmus tests,
  as well as for an adaptation of Peterson's algorithm for mutual exclusion
  to  causally consistent memory.
\end{abstract}

% \section{Structure - 18pp max - excluding bibliography for CAV}

% Assertional RG for potentials

% 1. Introduction 

% 2. Motivation: 
%     - Examples = Message passing, two RMWs to explain read-only potential, 
%     - Explain potentials

% 3. Background: Syntax + semantics of programs and memory systems (Current 3 + 4) 

% 4. Potential-based memory system and our interval logic 
%      - including algebraic laws of the logic 
     
% 5. Assertional RG 
%    - including the axioms for $\{P\} s \{Q\}$ (current section 10) 

% 6. Further Examples:
%    - Peterson's (TODO: Make critical section visible)
%    - for this: include swap in our syntax? 

\input{introduction}
\input{motivation}

\input{program}
\input{memory}

\input{og-general}
%\input{potential}
\input{potential2}
%\input{logic} 
\input{logic2} 
%\input{sra-sranew}
\input{axioms}

\input{examples}

\input{peterson}
\input{related}
% \input{conclusions}
%\input{rg-general}
% \input{notes}

%\bibliographystyle{alpha}
\bibliography{references}

\newpage
\appendix

\input{appendix}

\end{document}

%% file: defs.tex
\newcommand{\cmark}{{\color{green!70!black}\text{\ding{51}}}}%
\newcommand{\xmark}{{\color{red!70!black}\text{\ding{55}}}}%
\newcommand{\checkmarkt}[1]{%
	\edef\TVALUE{{#1}}%
	\expandafter\ifstrequal\TVALUE{yes}{\cmark}{}%
	\expandafter\ifstrequal\TVALUE{no}{\xmark}{}%
			\expandafter\ifstrequal\TVALUE{no*}{\xmark $^\star$}{}%
}

\spnewtheorem{cremark}{Remark}{\bfseries}{\itshape}
\spnewtheorem{cnotation}[definition]{Notation}{\bfseries}{\itshape}

\crefformat{section}{#2\S{}#1#3}
\Crefname{section}{Section}{Sections}
\Crefformat{section}{Section #2#1#3}
\crefname{corollary}{\text{Corollary}}{\text{Corollaries}}
\Crefname{corollary}{\text{Corollary}}{\text{Corollaries}}
\crefname{lemma}{\text{Lemma}}{\text{Lemmas}}
\Crefname{lemma}{\text{Lemma}}{\text{Lemmas}}
\crefname{proposition}{\text{Prop.}}{\text{Propositions}}
\Crefname{proposition}{\text{Proposition}}{\text{Propositions}}
\crefname{definition}{\text{Def.}}{\text{Definitions}}
\Crefname{definition}{\text{Definition}}{\text{Definitions}}
\crefname{notation}{\text{Notation}}{\text{Notations}}
\Crefname{notation}{\text{Notation}}{\text{Notations}}
\crefname{theorem}{\text{Thm.}}{\text{Theorems}}
\Crefname{theorem}{\text{Theorem}}{\text{Theorems}}
\crefname{figure}{\text{Fig.}}{\text{Figures}}
\Crefname{figure}{\text{Figure}}{\text{Figures}}
\crefname{example}{\text{Ex.}}{\text{Examples}}
\Crefname{example}{\text{Example}}{\text{Examples}}
\newcommand{\noop}[1]{}
\newcommand{\textcode}[1]{\texorpdfstring{\ensuremath{\mathtt{#1}}}{#1}}
\newcommand{\kw}[1]{\textbf{\textcode{#1}}}
\newcommand{\skipc}{\kw{skip}}
\newcommand{\ite}[3]{\kw{if}\;#1\:\kw{then}\;#2\; \kw{else}\;#3}

\newcommand{\while}[2]{\kw{while}\;#1\;\kw{do}\;#2}
\newcommand{\sep}{\;\kw{;}\;}
\newcommand{\spar}[3]{{\phantom{|}}^{#1}\!\kw{|}{#2}\kw{|}^{#3}\;}
\newcommand{\ALT}{\;\;|\;\;}

\newcommand{\ie}{{i.e.,} }
\newcommand{\eg}{{e.g.,} }
\newcommand{\etal}{{et~al.}}
\newcommand{\wrt}{w.r.t.~}
\newcommand{\aka}{a.k.a.~}

\newcommand{\cR}{\mathcal{R}} 
\newcommand{\cG}{\mathcal{G}}
\newcommand{\statesetR}{R}

\newcommand{\inarrT}[1]{\begin{array}[t]{@{}l@{}}#1\end{array}}

\newcommand{\inpar}[1]{\left(\begin{array}{@{}l@{}}#1\end{array}\right)}
\newcommand{\inset}[1]{\left\{\begin{array}{@{}l@{}}#1\end{array}\right\}}
\newcommand{\inarr}[1]{\begin{array}{@{}l@{}}#1\end{array}}

\DeclareRobustCommand{\brkstore}{\genfrac[]{0pt}{}}

\newcommand{\set}[1]{\{{#1}\}}
\newcommand{\sem}[1]{\llbracket #1 \rrbracket}
\newcommand{\pfn}{\rightharpoonup}

\newcommand{\st}{\; | \;}

\newcommand{\dom}[1]{\textit{dom}{({#1})}}

\newcommand{\fv}{\mathit{fv}}
\newcommand{\tup}[1]{{\langle{#1}\rangle}}
\newcommand{\nin}{\not\in}
\newcommand{\suq}{\subseteq}
\newcommand{\sqsuq}{\sqsubseteq}

\newcommand{\size}[1]{|{#1}|}

\newcommand{\true}{\mathit{True}}
\newcommand{\btrue}{\mathit{true}}
\newcommand{\false}{\mathit{False}}
\newcommand{\bfalse}{\mathit{false}}
\newcommand{\epsl}{\varepsilon}

\newcommand{\maketil}[1]{{#1}\ldots{#1}}
\newcommand{\til}{\maketil{,}}

\newcommand{\cdottil}{\cdot\!\ldots\!\cdot}

\newcommand{\rst}[1]{|_{#1}}

\newcommand{\defiff}{\mathrel{\stackrel{\mathsf{\triangle}}{\Leftrightarrow}}}
\newcommand{\defeq}{\triangleq}
\newcommand{\powerset}[1]{\mathcal{P}({#1})}

\renewcommand{\implies}{\Rightarrow}

\makeatletter
\newcommand{\raisemath}[1]{\mathpalette{\raisem@th{#1}}}
\newcommand{\raisem@th}[3]{\raisebox{#1}{$#2#3$}}
\makeatother

\newcommandx{\yaHelper}[2][1=\empty]{%
\ifthenelse{\equal{#1}{\empty}}%
  { \ensuremath{ \scriptstyle{ #2 } } } % no offset
  { \raisebox{ #1 }[0pt][0pt]{ \ensuremath{ \scriptstyle{ #2 } } } }  % with offset
}

\newcommandx{\yrightarrow}[4][1=\empty, 2=\empty, 4=\empty, usedefault=@]{%
  \ifthenelse{\equal{#2}{\empty}}
  { \xrightarrow{ \protect{ \yaHelper[ #4 ]{ #3 } } } } % there's no text below
  { \xrightarrow[ \protect{ \yaHelper[ #2 ]{ #1 } } ]{ \protect{ \yaHelper[ #4 ]{ #3 } } } } % there's text below
}

\newcommand{\astep}[1]{\mathrel{\raisebox{-0.8pt}{\ensuremath{\xrightarrow{#1}}}}}

\colorlet{colorPO}{gray!60!black}
\colorlet{colorRF}{green!60!black}
\colorlet{colorMO}{orange}
\colorlet{colorFR}{purple}
\colorlet{colorECO}{red!80!black}
\colorlet{colorSYN}{green!40!black}
\colorlet{colorHB}{blue}
\colorlet{colorPPO}{magenta}
\colorlet{colorPB}{olive}
\colorlet{colorSBRF}{olive}
\colorlet{colorRMW}{olive!70!black}
\colorlet{colorRSEQ}{blue}
\colorlet{colorSC}{violet}
\colorlet{colorPSC}{violet}
\colorlet{colorREL}{olive}
\colorlet{colorCONFLICT}{olive}
\colorlet{colorRACE}{olive}
\colorlet{colorWB}{orange!70!black}
\colorlet{colorPSC}{violet}
\colorlet{colorSCB}{violet}
\colorlet{colorDEPS}{violet}
\colorlet{colorS}{orange!70!black}

\tikzset{
   every path/.style={>=stealth},
   po/.style={->,color=colorPO,thin,shorten >=-0.5mm,shorten <=-0.5mm},
   sw/.style={->,color=colorSYN,shorten >=-0.5mm,shorten <=-0.5mm},
   rf/.style={->,color=colorRF,dashed,,shorten >=-0.5mm,shorten <=-0.5mm},
   hb/.style={->,color=colorHB,thick,shorten >=-0.5mm,shorten <=-0.5mm},
   mo/.style={->,color=colorMO,dotted,very thick,shorten >=-0.5mm,shorten <=-0.5mm},
   no/.style={->,dotted,thick,shorten >=-0.5mm,shorten <=-0.5mm},
   fr/.style={->,color=colorFR,dotted,thick,shorten >=-0.5mm,shorten <=-0.5mm},
   deps/.style={->,color=colorDEPS,dotted,thick,shorten >=-0.5mm,shorten <=-0.5mm},
   vis/.style={->,color=colorDEPS,dotted,thick,shorten >=-0.5mm,shorten <=-0.5mm},
   rmw/.style={->,color=colorRMW,thick,shorten >=-0.5mm,shorten <=-0.5mm},
}

\newcommand{\rlab}[3]{{\lR}^{#1}({#2},{#3})}

\newcommand{\wlab}[3]{{\lW}^{#1}({#2},{#3})}

\newcommand{\ulab}[4]{{\lU}^{#1}({#2},{#3},{#4})}

\newcommand{\forklab}[2]{{\lFORK}({#1},{#2})}
\newcommand{\joinlab}[2]{{\lJOIN}({#1},{#2})}
\newcommand{\forkcmd}[2]{{\lFORK}({#1},{#2})}
\newcommand{\joincmd}[2]{{\lJOIN}({#1},{#2})}

\newcommand{\lR}{{\mathtt{R}}}
\newcommand{\lW}{{\mathtt{W}}}

\newcommand{\lU}{{\mathtt{RMW}}}

\newcommand{\lFORK}{{\mathtt{FORK}}}
\newcommand{\lJOIN}{{\mathtt{JOIN}}}

\newcommand{\lQ}{{\mathtt{Q}}}

\newcommand{\lTheta}{{\mathbf{\Theta}}}

\newcommand{\lTID}{{\mathtt{tid}}}
\newcommand{\lVALO}{{\mathtt{val}}}
\newcommand{\lUFLAG}{{\mathtt{rmw}}}

\newcommand{\Tid}{\mathsf{Tid}}
\newcommand{\Loc}{\mathsf{Loc}}
\newcommand{\Val}{\mathsf{Val}}
\newcommand{\Lab}{\mathsf{Lab}}
\newcommand{\Reg}{\mathsf{Reg}}
\newcommand{\inter}{I}
\newcommand{\potassert}[2]{{#1} \!\ltimes\! #2 }
\newcommand{\gassert}{\varphi}
\newcommand{\chop}{\mathbin{;}}

\newcommand{\writeInstn}{\kw{STORE}}
\newcommand{\wrtInstn}{\kw{WRITE}}
\newcommand{\readInstn}{\kw{LOAD}}
\newcommand{\casInstn}{\kw{CAS}}
\newcommand{\swapInstn}{\kw{SWAP}}
\newcommand{\faddInstn}{\kw{FADD}}

\newcommand{\readInst}[2]{#1 \;{:=}\;\readInstn({#2})}
\newcommand{\writeInst}[2]{\writeInstn(#1,#2)}
\newcommand{\wrtInst}[2]{\wrtInstn(#1,#2)}
\newcommand{\assignInst}[2]{#1\;{:=}\;#2}

\newcommand{\swapInst}[2]{\swapInstn({#1},{#2})}

\newcommand{\makemodel}[1]{\ensuremath{{\mathsf{#1}}}\xspace}
\newcommand{\SRA}{\makemodel{SRA}}
\newcommand{\SC}{\makemodel{SC}}

\newcommand{\lora}{\makemodel{Piccolo}}

\newcommand{\makeL}[1]{{\ensuremath{\mathsf{lo}}{#1}}\xspace}
\newcommand{\SRAL}{\makeL{\SRA}}
\newcommand{\msc}{m}
\newcommand{\memstate}{{m}}

\renewcommand{\L}{L}

\newcommand{\D}{D}

\newcommand{\DD}{{\mathcal{{D}}}}
\newcommand{\LL}{{\mathcal{{L}}}}

\newcommand{\eexp}{E}
\newcommand{\potstore}{\delta}
\newcommand{\potstores}{\mathsf{\Delta}}

\newcounter{mylabelcounter}

\makeatletter
\newcommand{\labelAxiom}[2]{%
\hfill{\normalfont\textsc{(#1)}}\refstepcounter{mylabelcounter}
\immediate\write\@auxout{%
  \string\newlabel{#2}{{\unexpanded{\normalfont\textsc{#1}}}{\thepage}{{\unexpanded{\normalfont\textsc{#1}}}}{mylabelcounter.\number\value{mylabelcounter}}{}}
}%
}
\makeatother

\newcommand{\squishlist}[1][$\bullet$]{%[$\tinybullet$]{
 \begin{list}{#1}
  { \setlength{\itemsep}{0pt}
     \setlength{\parsep}{0pt}
     \setlength{\topsep}{1pt}
     \setlength{\partopsep}{0pt}
     \setlength{\leftmargin}{1.2em}
     \setlength{\labelwidth}{0.5em}
     \setlength{\labelsep}{0.4em} } }
\newcommand{\squishend}{
  \end{list}  }

\definecolor{DarkGreen}{rgb}{0.05, 0.45, 0.05}

\newcommand{\M}{\mathcal{M}}

\newcommand{\loc}{{x}}
\newcommand{\loca}{{y}}

\newcommand{\tid}{{\tau}}
\newcommand{\tida}{{\pi}}

\newcommand{\lab}{{l}}
\newcommand{\mlab}{{\theta}}
\newcommand{\val}{v}
\newcommand{\valr}{\val_\lR}
\newcommand{\valw}{\val_\lW}
\newcommand{\uflag}{u}

\newcommand{\reg}{{r}}

\renewcommand{\exp}{{e}}

\newcommand{\regstore}{{\gamma}}
\newcommand{\Regstores}{{\mathsf{\Gamma}}}
\newcommand{\regset}{Z}

\newcommand{\pcmd}{c}
\newcommand{\ipcmd}{\tilde{c}}
\newcommand{\cmd}{C}
\newcommand{\cmdmap}{\mathcal{C}}

\newcommand{\tidlab}[2]{\tup{{#1}:{#2}}}
\newcommand{\cmdstep}[2]{{}\mathrel{\raisebox{-0.8pt}{\ensuremath{\xrightarrow{#1,#2}}}}{}}

\newcommand{\asteptidlab}[3]{{}\mathrel{\raisebox{-0.8pt}{\ensuremath{\xrightarrow{{#1},{#2}}}}_{#3}}{}}
\newcommand{\asteplab}[2]{{}\mathrel{\raisebox{-0.8pt}{\ensuremath{\xrightarrow{#1}}}_{#2}}{}}

\makeatletter
\newcommand{\mylabel}[2]{#2\def\@currentlabel{#2}\label{#1}}
\makeatother

\newcommand{\ctid}[1]{\mathtt{T}_#1}
\newcommand{\cloc}[1]{\mathtt{#1}}
\newcommand{\creg}[1]{\mathtt{#1}}
\newcommand{\clist}[1]{\mathtt{#1}} 

  {\list{}{\leftmargin=#1\rightmargin=#1}\item[]}%
  {\endlist}
  
\newcommand{\seq}{\mathbin{;}}
\newcommand{\cs}[1]{\overline{#1}}

\newcommand{\satm}{\ \underline{sat}_\M\ }
\newcommand{\satsra}{\ \underline{sat}_\SRA\ }
 
\newcommand{\assume}{\mathsf{Assume}}
\newcommand{\commit}{\mathsf{Commit}}

\newcommand{\comp}{\xi}
\newcommand{\Comp}{\mathit{Comp}}

\newcommand{\pre}{P}
\newcommand{\post}{Q}
\newcommand{\pref}{\varphi}
\newcommand{\postf}{\psi}
\newcommand{\mida}{R}

\newcommand{\linefill}{\cleaders\hbox{$\smash{\mkern-2mu\mathord-\mkern-2mu}$}\hfill\vphantom{\lower1pt\hbox{$\rightarrow$}}}  
  
\newcommand{\transi}[2]{\mathrel{\lower1pt\hbox{$\mathrel-_{\vphantom{#2}}\mkern-8mu\stackrel{#1}{\linefill_{\vphantom{#2}}}\mkern-11mu\rightarrow_{#2}$}}}
\newcommand{\trans}[1]{\transi{#1}{{}}}

\newcommand{\ntransi}[2]{\mathrel{\lower1pt\hbox{$\mathrel-_{\vphantom{#2}}\mkern-8mu\stackrel{#1}{\linefill_{\vphantom{#2}}}\mkern-8mu\nrightarrow_{#2}$}}}

\newcommand{\astate}{\sigma}
\newcommand{\allstates}{\Sigma}
\newcommand{\compstep}{\mathtt{cmp}}
\newcommand{\envstep}{\mathtt{env}}
\newcommand{\memstep}{\mathtt{mem}}

\newcommand{\last}{\mathsf{last}}

\newcommand{\wf}{\mathsf{wf}}

\newcommand{\rem}[2]{\mathsf{rem}(#1,#2)} 
\newcommand{\ind}[2]{\mathsf{ind}(#1,#2)} 

\newcommand{\prop}{{P}}

\newcommand{\propg}{{G}}

\newcommand{\assert}[1]{{\color{blue}{\ensuremath{\left\{
    \begin{array}[c]{@{}l@{}}
      #1
    \end{array}
\right\}}}}}

\newcommand{\sassert}[1]{{\color{blue}{\ensuremath{\{
      #1
\}}}}}

\newcommand{\mhoare}[3]{\{#1\}\;{#2}\;\{#3\}}

\newcommand{\ahoare}[4]{\sassert{#1} \; {#2} \; \sassert{#3}}
\newcommand{\srahoare}[3]{\ahoare{#1}{#2}{#3}{\satsra}}

\newcommand{\gc}[2]{\{#1\} \; #2}
\newcommand{\subst}[2]{(#1:=#2)}

%% file: introduction.tex
\section{Introduction}

Rely-guarantee (RG) is a fundamental compositional 
proof technique for concurrent programs~\cite{DBLP:journals/toplas/Jones83,DBLP:journals/fac/XuRH97}. 
Each program component $P$ is specified using {\em rely} and {\em guarantee} conditions,
which means that $P$ can tolerate any environment interference that follows its rely condition,
and generate only interference included in its guarantee condition.
Two components can be composed in parallel provided that the rely of each component
agrees with the guarantee of the other.

The original RG framework and its soundness proof
have assumed a sequentially consistent (SC) memory~\cite{DBLP:journals/tc/Lamport79},
which is unrealistic in modern processor architectures and programming languages.
Nevertheless, the main principles behind RG are not at all specific for SC.
Accordingly, our first main contribution, is to formally decouple the
underlying memory model from the RG proof principles,
by proposing a generic RG framework parametric in the input memory model.
To do so, we assume that the underlying memory model
is axiomatized by Hoare triples specifying pre- and postconditions on memory 
states for each primitive operation (\eg loads and stores).
This enables the formal development of RG-based logics for different
shared memory models as instances of one framework, where all build on a uniform
soundness infrastructure of the RG rules (\eg for sequential and parallel composition), but employ different specialized
assertions to describe the possible memory states,
where specific soundness arguments are only needed for primitive memory operations.

The second contribution of this paper is an instance of the general RG framework
for \emph{causally consistent shared memory}.
The latter stands for a family of wide-spread and well-studied memory models weaker than SC,
which are sufficiently strong for implementing a variety of synchronization idioms~\cite{DBLP:journals/siglog/Lahav19,DBLP:journals/dc/AhamadNBKH95,DBLP:conf/popl/BouajjaniEGH17}.
Intuitively, unlike SC, causal consistency allows different threads to observe writes to memory in different orders,
as long as they agree on the order of writes that are causally related.
This concept can be formalized in multiple ways, and here
 we target a strong form of causal consistency, called \emph{strong release-acquire} (SRA)~\cite{DBLP:conf/popl/LahavGV16,DBLP:journals/toplas/LahavB22}
(and equivalent to ``causal convergence'' from \cite{DBLP:conf/popl/BouajjaniEGH17}), 
which is a slight strengthening of the well-known release-acquire (RA) model (used by C/C++11). 
(The variants of causal consistency only differ for programs with write/write races~\cite{DBLP:journals/toplas/LahavB22,DBLP:journals/lmcs/BeillahiBE21},
which are rather rare in practice.)

%WHY SRA?

Our starting point for axiomatizing SRA as Hoare triples is the
\emph{potential-based} operational semantics of SRA, which was
recently introduced with the goal of establishing the decidability of
control state reachability under this model~\cite{DBLP:conf/pldi/LahavB20,DBLP:journals/toplas/LahavB22}
(in contrast to undecidability under
RA~\cite{DBLP:conf/pldi/AbdullaAAK19}).  Unlike more standard
presentations of weak memory models whose states record information
about the \emph{past} 
(\eg in the form of store buffers containing executed writes before they are globally visible~\cite{DBLP:conf/tphol/OwensSS09}, 
partially ordered execution graphs~\cite{DBLP:conf/popl/LahavGV16,DBLP:journals/toplas/AlglaveMT14,DBLP:conf/ppopp/DohertyDWD19},
or collections of timestamped messages and thread
views~\cite{DBLP:conf/ecoop/KaiserDDLV17,DBLP:conf/popl/KangHLVD17,DBLP:conf/ecoop/DalvandiDDW19,DBLP:journals/jar/DalvandiDDW22,DBLP:conf/fm/WrightBD21,DBLP:conf/esop/BilaDLRW22}),
the states of the potential-based model track possible \emph{futures}
ascribing what sequences of observations each thread can perform. 
We find this approach to be a particularly appealing candidate 
for Hoare-style reasoning which would naturally generalize SC-based reasoning.
Intuitively, while an assertion in SC specifies possible observations at a given program point,
an assertion in a potential-based model should specify possible \emph{sequences} of observations.

To pursue this direction, we introduce a novel assertion language, resembling temporal logics, which allows one to express properties of sequences of states.
For instance, our assertions can express that a certain thread may currently read $\cloc{x}=0$,
 but it will have to read $\cloc{x}=1$ once it reads $\cloc{y}=1$.
Then, we provide Hoare triples for  SRA in this assertion language, and incorporate them in the general RG framework.
The resulting program logic, which we call \lora, provides a novel approach to reason on concurrent programs under causal consistency,
which allows for simple and direct proofs, and, we believe, 
may constitute a basis for automation in the future.

% The rest of this paper is organized as follows: \todo{}

% \paragraph{Contributions.}
% \begin{enumerate}
% \item 
% \end{enumerate}

% \paragraph{Overview.}
% This paper is structured as follows.

% \begin{itemize}
% \item We choose to apply the framework on a potential-based memory system, which is a novel approach to define weak memory models
% % \item More concretely, we target \SRAL, a potnetial-based system for  strong release-acquire (SRA), which is a
% %   strengthening of C11's release-acquire (RA) that requires that ordering of writes to {\em all} locations is
% %   consistent with happens-before order.
% \end{itemize}

%%% Local Variables:
%%% mode: latex
%%% TeX-master: "main"
%%% End:

%% file: motivation.tex
\section{Motivating Example} \label{sec:motivation}

To make our discussion concrete, consider the message passing program (MP) in \cref{fig:mp-SC,fig:mp-SRA}, 
comprising shared variables $\cloc{x}$ and $\cloc{y}$ and local registers $\creg{a}$ and $\creg{b}$. 
The proof outline in \cref{fig:mp-SC} assumes SC, whereas \cref{fig:mp-SRA} assumes SRA. 
In both cases, at the end of the execution, we show that if $\creg{a}$ is $1$, then $\creg{b}$ must also be $1$. 
We use these examples to explain the two main concepts introduced in this paper: 
$(i)$ a generic RG framework and 
$(ii)$ its instantiation with a potential-focused assertion system that enables reasoning under SRA.

\begin{figure}[t]
  \centering 
  \begin{minipage}[b]{0.42\columnwidth}
  \scalebox{1}{
  \begin{tabular}[b]{@{}c@{}} 
    $\assert{ \cloc y \neq 1    } $ \\ 
  $\begin{array}{@{}l@{~}||@{~}l@{}}
    \begin{array}[t]{l}
     \textbf{Thread } \ctid{1}
     \\
     \assert{ \true  } \\ %G_1(\cloc{y} \to 0)
     1: \writeInst{\cloc{x}}{1}; \\ 
      \assert{ \cloc{x} = 1  } \\ % \wedge y \to 0
     2: \writeInst{\cloc{y}}{1} \\
     \assert{ \true   } 
     \end{array}
    & 
    \begin{array}[t]{l}
     \textbf{Thread } \ctid{2}
     \\
      \assert{ \cloc{y} = 1 \implies \cloc{x} = 1    } \\
     3: \readInst{\creg{a}}{\cloc{y}} ;
     \\
      \assert{ \creg{a} = 1  \implies \cloc x = 1  } \\ 
     4: \readInst{\creg{b}}{\cloc{x}} \\ 
     \assert{ \creg{a} = 1 \implies \creg{b} = 1 } 
     \end{array}
     \end{array}$ \\
     $\assert{ \creg{a}=1 \implies \creg{b}=1    } $
  \end{tabular}}
%  \vspace{-5pt}
 \caption{Message passing in SC}
 \label{fig:mp-SC}
  \end{minipage}
  \hfill
  \begin{minipage}[b]{0.52\columnwidth}
  \scalebox{1}{
  \begin{tabular}[b]{@{}c@{}} 
    $\assert{\potassert{\ctid{0}}{[ \cloc{y} \neq 1]}} $ \\ 
  $\begin{array}{@{}l@{~}||@{~}l@{}}
    \begin{array}[t]{l}
     \textbf{Thread } \ctid{1}
     \\
     \assert{ \true  } \\ %G_1(y \to 0)
     1: \writeInst{\cloc{x}}{1}; \\ 
      \assert{
      \potassert{\ctid{1}}{[ \cloc{x}=1]}} \\ % \wedge y \to 0
     2: \writeInst{\cloc{y}}{1} \\
     \assert{ \true   } 
     \end{array}
    & 
    \begin{array}[t]{l}
     \textbf{Thread } \ctid{2}
     \\
      \assert{
      \potassert{\ctid{2}}{[\cloc{y}\neq 1]\chop[ \cloc{x}=1]}}
      \\

      3: \readInst{\creg{a}}{\cloc{y}} ;
     \\
      \assert{ \creg{a} = 1  \implies
      \potassert{\ctid{2}}{[ \cloc{x}=1]}
      % \pot{2} \in  [\cloc{x} \mapsto 1]
      } \\ 
     4: \readInst{\creg{b}}{\cloc{x}} \\ 
     \assert{ \creg{a} = 1 \implies \creg{b} = 1 } 
     \end{array}
     \end{array}$ \\
     $\assert{\creg{a}=1 \implies \creg{b}=1    } $
  \end{tabular} }
%\vspace{-5pt}
 \caption{Message passing in SRA}
 \label{fig:mp-SRA}
  \end{minipage}

\end{figure}

\paragraph{\textbf{Rely-Guarantee.}}
The proof outline in \cref{fig:mp-SC} can be % easily
read as an RG derivation:
\begin{enumerate}[noitemsep,topsep=2pt]
\item Thread $\ctid{1}$ locally establishes its postcondition when starting from any state that satisfies its precondition.
This is trivial since its postcondition is $\true$.
\item Thread $\ctid{1}$ relies on the fact that its used assertions are \emph{stable} \wrt interference from its environment.
We formally capture this condition by a rely set $\cR_1 \defeq \set{\true, \cloc{x} = 1}$.
\item Thread $\ctid{1}$ guarantees to its concurrent environment that its only interferences are 
$\writeInst{\cloc{x}}{1}$ and $\writeInst{\cloc{y}}{1}$, and furthermore that $\writeInst{\cloc{y}}{1}$ is only performed when $\cloc{x} = 1$ holds.
We formally capture this condition by a guarantee set $\cG_1 \defeq \set{\gc{\true}{\ctid{1} \mapsto \writeInst{\cloc{x}}{1}},
  \gc{\cloc{x} = 1}{\ctid{1} \mapsto \writeInst{\cloc{y}}{1}}} $, where each element is a command guarded by a precondition.
\item Thread $\ctid{2}$ locally establishes its postcondition when starting from any state that satisfies its precondition.
This is straightforward using standard Hoare rules for assignment and sequential composition.
\item Thread $\ctid{2}$'s rely set is again obtained by collecting all the assertions used in its proof:
 $\cR_2 \defeq \set{\cloc{y} = 1 \implies \cloc{x} = 1, \creg{a} = 1 \implies \cloc{x} = 1, \creg{a} = 1 \implies \creg{b} = 1}$.
 Indeed, the local reasoning for $\ctid{2}$ needs all these assertions to be stable under the environment interference.
\item Thread $\ctid{2}$'s guarantee set is given by: %\hw{use disjunction here as in proof outline and the rely of 2?} 
$$\small \cG_2 \defeq \inset{\gc{\cloc{y} = 1 \implies \cloc{x} = 1}{\ctid{2} \mapsto \readInst{\creg{a}}{\cloc{y}}},
    \gc{\creg{a} = 1 \implies \cloc{x} = 1}{\ctid{2} \mapsto
      \readInst{\creg{b}}{\cloc{x}}}} $$
\item To perform the  parallel composition,
$\tup{\cR_1, \cG_1}$ and $\tup{\cR_2, \cG_2}$ should be
\emph{non-interfering}. This involves showing that each $R \in \cR_i$ is
\emph{stable} under each $G \in \cG_j$ for $i \neq j$. That is, if
$G = \gc{P}{\tid \mapsto c}$, we require
the Hoare triple $\mhoare{P \cap R}{\tid \mapsto c}{R}$ to hold.
% which shows that under the
%precondition $P$, the command $c$ executed by $\tid$ ensures stability of $R$. 
In this case, these proof obligations are straightforward to
discharge using Hoare's assignment axiom 
(and is trivial for $i=1$ and $j=2$ since
 load instructions leave the memory intact). \label{non-interfering}
\end{enumerate}

\begin{remark}
  Classical treatments of RG involve two related
  ideas~\cite{DBLP:journals/toplas/Jones83}: (1) specifying a
  component by rely and guarantee conditions (together with standard
  pre- and postconditions); and (2) taking the relies and guarantees
  to be binary relations over states.  Our approach adopts (1) but not (2).
%  \bdc{avoids
%    (2)}{maybe better to day that we do a more restrictive form of
%    (2), e.g., ``adopt a more restricted form of (2), where a rely
%    only describes assertions that are must be stable under
%    interference.''?}.  
Thus, it can be seen as an RG
  presentation of the Owicki-Gries method~\cite{DBLP:journals/acta/OwickiG76}, as was previously
  done in~\cite{DBLP:conf/icalp/LahavV15}. We have not observed an
  advantage for using binary relations in our examples, but the framework can be
  straightforwardly modified to do so.
\end{remark}

Now, observe that substantial aspects of the above reasoning are
\emph{not} directly tied with SC.  This includes the Hoare rules for
compound commands (such as
sequential composition above), the idea of specifying a thread using
collections of stable rely assertions and guaranteed guarded primitive
commands, and the non-interference condition for parallel composition.
To carry out this generalization, we assume that we are provided an
assertion language whose assertions are interpreted as \emph{sets of memory
states} (which can be much more involved than simple mappings of
variables to values), and a set of valid Hoare triples for the
primitive instructions.  The latter is used for checking validity of
primitive triples, (\eg
$\mhoare{P}{\ctid{1} \mapsto \writeInst{\cloc{x}}{1}}{Q}$), as well as
non-interference conditions (\eg
$\mhoare{P \cap R}{\ctid{1} \mapsto \writeInst{\cloc{x}}{1}}{R}$).  In
\cref{sec:genRG}, we present this generalization, and establish the
soundness of RG principles independently of the memory model.

\paragraph{\textbf{Potential-based reasoning.}}
The second contribution of our work is an application of the above to
develop a logic for a potential-based operational semantics that captures SRA.
In this semantics every memory state records sequences of store mappings
(from shared variables to values) that each thread may observe.
For example, assuming all variables are initialized to $0$, 
if $\ctid{1}$ executed its code until completion before $\ctid{2}$ even started
(so under SC the memory state is the store $\set{\cloc{x}\mapsto 1, \cloc{y} \mapsto 1}$),
we may reach the SRA state in which
$\ctid{1}$'s potential consists of one store $\set{\cloc{x}\mapsto 1, \cloc{y} \mapsto 1}$,
and $\ctid{2}$'s potential is the sequence of stores:
$$\tup{\set{\cloc{x}\mapsto 0, \cloc{y} \mapsto 0}, 
\set{\cloc{x}\mapsto 1, \cloc{y} \mapsto 0},
\set{\cloc{x}\mapsto 1, \cloc{y} \mapsto 1}},$$
which captures the stores that $\ctid{2}$ may observe 
in the order it may observe them.  
Naturally, potentials are \emph{lossy}
allowing threads to non-deterministically lose a subsequence of the current store sequence,
so they can progress in their sequences.
Thus, $\ctid{2}$ can read $1$ from $\cloc{y}$ only after it loses the first two stores in its potential,
and from this point on it can only read $1$ from $\cloc{x}$.
Now, one can see that \emph{all} potentials of $\ctid{2}$ at its initial program point
are, in fact, subsequences of the above sequence (regardless of where $\ctid{1}$ is),
and conclude that $\creg{a} = 1 \implies \creg{b} = 1$ holds when $\ctid{2}$ terminates.

To capture the above informal reasoning in a Hoare logic, 
we designed a new form of assertions capturing possible locally observable sequences of stores, rather than one global store,
which can be seen as a restricted fragment of linear temporal logic.
The proof outline using these assertions is given in \cref{fig:mp-SRA}.
In particular, $[\cloc{x} = 1]$ is satisfied by all store sequences in which every store maps $\cloc{x}$ to $1$,
whereas $[\cloc{y} \neq 1] \chop [\cloc{x} = 1]$  is satisfied by all store sequences
that can be split into a (possibly empty) prefix whose value for $\cloc{y}$ is not $1$
followed by a (possibly empty) suffix whose value for $\cloc{x}$ is $1$.
Assertions of the form $\potassert{\tid}{\inter}$ state that the potential of thread $\tid$
includes only store sequences that satisfy $\inter$.

The first assertion of $\ctid{2}$ is implied by the initial condition,
$\potassert{\ctid{0}}{[ \cloc{y} \neq 1]}$, since the potential of the
parent thread $\ctid{0}$ is inherited by the forked child threads and 
$\potassert{\ctid{2}}{[ \cloc{y} \neq 1]}$ implies
$\potassert{\ctid{2}}{[ \cloc{y} \neq 1] \chop \inter}$ for any $\inter$.
%since chop
%$;$ can be applied so that $[ \cloc{y} \neq 1]$ encompases the entire original interval. 
Moreover,
$\potassert{\ctid{2}}{[ \cloc{y} \neq 1] \chop [\cloc{x} = 1]}$ is
preserved by
\begin{enumerate*}[label=(\roman*)]
%\begin{itemize}[noitemsep,topsep=2pt]
\item 
line 1 because writing 1 to $\cloc x$ leaves
$[\cloc{y} \neq 1]$ unchanged and re-establishes $[\cloc x = 1]$; and
\item 
 line 2 because the % potential-based
semantics for SRA ensures that after reading $1$ from $\cloc{y}$ by $\ctid{2}$, 
the thread $\ctid{2}$ is confined by $\ctid{1}$'s
potential just before it wrote $1$ to $\cloc{y}$,
which has to satisfy the precondition $\potassert{\ctid{1}}{[ \cloc{x}=1]}$.
(SRA allows to update the other threads' potential
only when the suffix of the potential after the update is observable by the writer thread.)
\end{enumerate*}

%   which ensures
%  that when the information $[\cloc{y}=1]$ is transferred to thread
%  $\ctid{2}$, so is $[\cloc{x}=1]$. \ori{this is important but unclear :(}
%\end{itemize}
% $\potassert{\ctid{1}}{[\cloc{x}=1]}$ since 

In \cref{sec:logic} we formalize these arguments
as Hoare rules for the primitive instructions,
whose soundness is checked using the potential-based 
operational semantics and the interpretation of the assertion language.
Finally, \lora is obtained by incorporating these Hoare rules in the general RG framework.

\begin{remark}
\label{rem:sra}
  Our presentation of the potential-based semantics for SRA (fully
  presented in \cref{sec:potential}) deviates from the original one in
  \cite{DBLP:journals/toplas/LahavB22}, where it was called \SRAL.  The most
  crucial difference is that while \SRAL's potentials consist of lists
  of per-location read options, our potentials consist of lists of
  \emph{stores} assigning a value to every variable.  (This is similar
  in spirit to the adaptation of load buffers for
  TSO~\cite{DBLP:conf/concur/AbdullaABN16,DBLP:journals/lmcs/AbdullaABN18}
  to snapshot buffers in~\cite{DBLP:journals/pacmpl/AbdullaABKS21}).
  Additionally, unlike \SRAL, we disallow empty potential lists,
  require that the potentials of the different threads agree on the
  very last value to each location, and handle read-modify-write (RMW)
  instructions differently.  We employed these modifications to \SRAL
  as we observed that direct reasoning on \SRAL states is rather
  unnatural and counterintuitive, as \SRAL allows traces that {\em
    block} a thread from reading any value from certain locations
  (which cannot happen in the version we formulate).  For example, a
  direct interpretation of our assertions over \SRAL states would
  allow states in which $\potassert{\tid}{[\loc=\val]}$ and
  $\potassert{\tid}{[\loc\neq\val]}$ both hold (when $\tid$ does not
  have any option to read from $\loc$), while these assertions are
  naturally contradictory when interpreted on top of our modified SRA
  semantics.  To establish confidence in the new potential-based
  semantics we have proved \emph{in Coq} its equivalence to the
  standard execution-graph based semantics of SRA (over 5K lines of
  Coq proofs)~\cite{CAV-Artifact}.
\end{remark}

%% file: program.tex
\section{Preliminaries: Syntax and Semantics}

\begin{figure}[t]
  \centering
  \begin{tabular}{rll@{\ \ \  }rll}
\emph{values} & $\ $ & $\val \in \Val = \set{0,1,\ldots}$ %\qquad 0 \in \Val$
&  \emph{shared variables} & $\ $ & $\loc, \loca \in \Loc = \set{\cloc{x},\cloc{y},\ldots}$
\\ \emph{local registers} & $\ $ & $\reg \in \Reg = \set{\creg{a},\creg{b},\ldots}$ % \regarg
&  \emph{thread identifiers} & $\ $ & $\tid, \tida \in \Tid = \set{\ctid{0},\ctid{1},\ldots}$
\end{tabular}
%\vspace{-5pt}

\[\begin{array}{@{} l l @{}}
\exp ::=  & \reg \ALT \val \ALT \exp + \exp \ALT \exp = \exp \ALT \neg \exp \ALT \exp \land \exp \ALT \exp \lor \exp \ALT \ldots
\\[0.5ex]
\pcmd ::=  &
\inarrT{\assignInst{\reg}{\exp} 
\ALT \writeInst{\loc}{\exp}
\ALT \readInst{\reg}{\loc}
% \ALT \faddInst{\reg}{\loc}{\exp}
% \ALT \casInst{\reg}{\loc}{\exp}{\exp} \\ 
\ALT \swapInst{\loc}{\exp}}
% \\[0.5ex]
\qquad\qquad \ipcmd ::=  
\tup{\pcmd,\assignInst{\vec{\reg}}{\vec{\exp}}}
\\[0.5ex]
\cmd ::=  &
\pcmd
\ALT \ipcmd
\ALT \skipc
\ALT \cmd \sep \cmd
\ALT \ite{\exp}{\cmd}{\cmd}
\ALT \while{\exp}{\cmd}
\ALT \cmd \spar{\tid}{}{\tid} \cmd 
\end{array}
\]
%\vspace{-15pt}
\caption{Program syntax}
\label{fig:syntax}
\end{figure}

In this section we describe the underlying program language,
leaving the shared-memory semantics parametric.

\paragraph{\textbf{Syntax.}}
The syntax of programs, given in \cref{fig:syntax}, is mostly
standard, comprising primitive (atomic) commands $\pcmd$ and compound
commands $\cmd$. The non-standard components are instrumented commands
$\ipcmd$, which are meant to atomically execute a primitive command $\pcmd$ and
a (multiple) assignment $\assignInst{\vec{\reg}}{\vec{\exp}}$.  Such instructions are needed to
support auxiliary (\aka ghost) variables in RG proofs.  In
addition, $\swapInstn$ (\aka atomic exchange) is an
example of an RMW instruction.  For brevity, other
standard RMW instructions, such as $\faddInstn$ and $\casInstn$, are
omitted.
% We denote by $\Exp$ and $\Cmd$ the set of all
% expressions and commands (respectively). \bd{is this notation used?}

% The domains and
% metavariables used to range over them are as follows:

% We use vector notation for finite sequences of elements 
% (\eg, $\vec{\reg}$ is standing for finite sequence of register names).

% Expressions, primitive commands, instrumented primitive commands, and commands:

%Main threads in a command:
%
%$$\lTID(\cmd) = \begin{cases}
%\lTID(\cmd_1) \cup \lTID(\cmd_2) & \cmd = \cmd_1 \sep \cmd_2 \\
%\lTID(\cmd_1) \cup \lTID(\cmd_2) & \cmd = \ite{\exp}{\cmd_1}{\cmd_2} \\
%\lTID(\cmd_1) & \cmd = \while{\exp}{\cmd_1} \\
%\set{\tid} & \cmd =  \cmd_1 \spar{\tid_1}{\tid}{\tid_2} \cmd_2 \\
%\emptyset & \text{otherwise}
%\end{cases}$$

Unlike many weak memory models that only support top-level parallelism,
%\footnote{See \cite{DBLP:conf/netys/AbdullaABKS22}  for an exception.}
we include dynamic thread creation via commands of the form
$\cmd_1 \spar{\tid_{1}}{}{\tid_{2}} \cmd_2$ that forks two threads named
$\tid_1$ and $\tid_2$ that execute the commands $\cmd_1$ and $\cmd_2$,
respectively. Each $\cmd_i$ may itself comprise further
parallel compositions. 
Since thread identifiers are explicit, we
require commands to be \emph{well formed}. Let $\Tid(\cmd)$ be the set
of all thread identifiers that appear in $\cmd$.
% $$\Tid(\cmd) = \begin{cases}
% \Tid(\cmd_1) \cup \Tid(\cmd_2) & \cmd = \cmd_1 \sep \cmd_2 \\
% \Tid(\cmd_1) \cup \Tid(\cmd_2) & \cmd = \ite{\exp}{\cmd_1}{\cmd_2} \\
% \Tid(\cmd_1) & \cmd = \while{\exp}{\cmd_1} \\
% \Tid(\cmd_1) \cup \Tid(\cmd_2) \cup \set{\tid_1,\tid_2} & \cmd =  \cmd_1 \spar{\tid_1}{}{\tid_2} \cmd_2 \\
% \emptyset & \text{otherwise}
% \end{cases}$$
A command $\cmd$ is \emph{well formed}, denoted $\wf(\cmd)$, if
parallel compositions inside employ disjoint sets of thread
identifiers.  This notion is formally defined by induction on the
structure of commands, with the only interesting case being
% parallel composition, where 
$\wf(\cmd_1 \spar{\tid_1}{}{\tid_2} \cmd_2)$ if
$\wf(\cmd_1) \wedge \wf(\cmd_2) \wedge \tid_1 \neq \tid_2 \wedge
\Tid(\cmd_1) \cap \Tid(\cmd_2) =\emptyset$.
% $$\inferrule*{\wf(\cmd_1) \\ \wf(\cmd_2) \\ 
% \tid_1 \neq \tid_2 \\ \Tid(\cmd_1) \cap \Tid(\cmd_2) =\emptyset
% }{\wf(\cmd_1 \spar{\tid_1}{}{\tid_2} \cmd_2)}$$
%Below we assume that all commands are well-formed.

%\begin{mathpar}
%\inferrule*{\primcmd \in \PrimCmd}
%{\wf(\primcmd,\tid)}
%\and 
%\inferrule*{\ }
%{\wf(\skipc,\tid)}
%\and 
%\inferrule*{\ }
%{\wf(\assignInst{\vec{\reg}}{\vec{\exp}},\tid)}
%\and 
%\inferrule*{\wf(\cmd_1,\tid) \\ \wf(\cmd_2,\tid)}
%{\wf({\cmd_1 \sep \cmd_2},\tid)}
%\and 
%\inferrule*{\wf(\cmd_1,\tid) \\ \wf(\cmd_2,\tid)}
%{\wf(\ite{\exp}{\cmd_1}{\cmd_2},\tid)}
%\and 
%\inferrule*{\wf(\cmd,\tid)}
%{\wf(\while{\exp}{\cmd},\tid)}
%\and
%\inferrule*{\wf(\cmd_1,\tid_1) \\ \wf(\cmd_2,\tid_2) \\\\
%\tid, \tid_1, \tid_2 \text{ are distinct} \\\\
%\Tid(\cmd_1) \cap \Tid(\cmd_2) =\emptyset \\\\ 
%\tid\nin \Tid(\cmd_1) \cup \Tid(\cmd_2)}
%{\wf(\cmd_1 \spar{\tid_1}{\tid}{\tid_2} \cmd_2,\tid)}
%\end{mathpar}

%We say that a command $\cmd$ is \emph{well-formed} 
%if it is well-formed \wrt $\ctid{0}$.

% \subsection{Semantics}

\newcommand{\pcmdsem}{
\begin{figure}[t]
\begin{mathpar}
\small
\inferrule*{
\regstore'=\regstore[\reg \mapsto \regstore(\exp)]
}{
\assignInst{\reg}{\exp} \gg \regstore \astep{\epsl} \regstore'
}
\and
\inferrule*{
\lab =\wlab{}{\loc}{\regstore(\exp)} \\
}{
\writeInst{\loc}{\exp} \gg \regstore \astep{\lab} \regstore
}
\and
\inferrule*{
\lab =\rlab{}{\loc}{\val} \\
\regstore'=\regstore[\reg \mapsto \val]
}{
\readInst{\reg}{\loc} \gg \regstore \astep{\lab} \regstore'
}
% \and
% % \inferrule*{
% % \lab =\ulab{}{\loc}{\val}{\val+\regstore(\exp)} \\\\
% % \regstore'=\regstore[\reg \mapsto \val]
% % }{
% % \faddInst{\reg}{\loc}{\exp} \gg \regstore \astep{\lab} \regstore'
% % }
% \and
% \inferrule*{
% \val = \regstore(\exp_\lR) \\
% \lab =\ulab{}{\loc}{\val}{\regstore(\exp_\lW)} \\\\
% \regstore'=\regstore[\reg \mapsto \val] 
% }{
% \casInst{\reg}{\loc}{\exp_\lR}{\exp_\lW} \gg \regstore \astep{\lab} \regstore'
% }
% \and
% \inferrule*{
% \val \neq \regstore(\exp_\lR) \\
% \lab =\rlab{}{\loc}{\val} \\\\
% \regstore'=\regstore[\reg \mapsto \val]
% }{
% \casInst{\reg}{\loc}{\exp_\lR}{\exp_\lW} \gg \regstore \astep{\lab} \regstore'
% }
\\
\inferrule*{
\lab =\ulab{}{\loc}{\val}{\regstore(\exp)} 
}{
\swapInst{\loc}{\exp} \gg \regstore \astep{\lab} \regstore
}
\and
\inferrule*{
\pcmd \gg \regstore \astep{\lab_\epsl} \regstore_0\\\\
\assignInst{\reg_1}{\exp_1} \gg \regstore_0 \astep{\epsl} \regstore_1\\ \ldots  \\
\assignInst{\reg_n}{\exp_n} \gg \regstore_{n-1} \astep{\epsl} \regstore_n
%\regstore''=\regstore'[\reg_1 \mapsto \regstore'(\exp_1)]\ldots[\reg_n \mapsto \regstore'(\exp_n)]
}{
\tup{\pcmd,\assignInst{\tup{\reg_1 \til \reg_n}}{\tup{\exp_1 \til \exp_n}}} \gg \regstore \astep{\lab_\epsl} \regstore_n
}
\end{mathpar}

%\vspace{-10pt}
\caption{Small-step semantics of (instrumented) primitive commands  ($\ipcmd \gg \regstore \astep{\lab_\epsl} \regstore'$)}
%\bd{Removed FADD and CAS}}
\label{fig:pcmdsem}
\end{figure}
}

\newcommand{\cmdsem}{
\begin{figure}[t]
\begin{mathpar}
\small
\inferrule*{
\ipcmd \gg \regstore \astep{\lab_\epsl} \regstore'
}{
\tup{\ipcmd,\regstore} \astep{\lab_\epsl} \tup{\skipc,\regstore'}
}
\and
\inferrule*{
\tup{\cmd_1,\regstore} \astep{\lab_\epsl} \tup{\cmd_1',\regstore'} 
}{
\tup{\cmd_1 \sep \cmd_2,\regstore} \astep{\lab_\epsl} \tup{\cmd_1' \sep \cmd_2,\regstore'}
}
\and
\inferrule*{
\ }{
\tup{\skipc \sep \cmd_2,\regstore} \astep{\epsl} \tup{\cmd_2,\regstore}
}
\and
\inferrule*{
\regstore(\exp) = \btrue \implies i =1 \\\\
\regstore(\exp) \neq\btrue \implies i=2 
}{
\tup{\ite{\exp}{\cmd_1}{\cmd_2},\regstore} \astep{\epsl} \tup{\cmd_i,\regstore}
}
\hfill
\inferrule*{
\cmd'=\ite{\exp}{(\cmd \sep \while{\exp}{\cmd})}{\skipc}
}{
\tup{\while{\exp}{\cmd},\regstore} \astep{\epsl} \tup{\cmd',\regstore}
}
\end{mathpar}
%\vspace{-20pt}
\caption{Small-step semantics of commands  ($\tup{\cmd,\regstore} \astep{\lab_\epsl} \tup{\cmd',\regstore'}$)}
\label{fig:cmdsem}
\end{figure}
}

\newcommand{\cmdpoolsem}{
\begin{figure}[t]
\scalebox{0.95}{\begin{mathpar}
\small
\inferrule*{
\tup{\cmd,\regstore} \astep{\lab_\epsl} \tup{\cmd',\regstore'}
}{
  \inarr{\tup{\cmdmap_0 \uplus \set{\tid\mapsto \cmd},\regstore} \\
    \cmdstep{\tid}{\lab_\epsl} \tup{\cmdmap_0 \uplus \set{\tid\mapsto \cmd'},\regstore}}
}
\quad
\inferrule*{
\cmdmap(\tid) = \cmd_1 \spar{\tid_1}{}{\tid_2} \cmd_2 \\\\
 \tid_1 \nin \dom{\cmdmap} \\ \tid_2 \nin \dom{\cmdmap} \\\\
\lab = \forklab{\tid_1}{\tid_2} \\\\
\cmdmap'= \set{\tid_1 \mapsto \cmd_1, \tid_2 \mapsto \cmd_2}
}{
\tup{ \cmdmap,\regstore} \cmdstep{\tid}{\lab} \tup{ \cmdmap \uplus \cmdmap',\regstore}
}
\quad
\inferrule*{
  \cmdmap= \inset{\tid\mapsto \cmd_1 \spar{\tid_1}{}{\tid_2} \cmd_2, \\
    \tid_1 \mapsto \skipc, \tid_2 \mapsto \skipc}\\\\
\lab = \joinlab{\tid_1}{\tid_2} \\\\
\cmdmap'= \set{\tid \mapsto \skipc}
}{
\tup{ \cmdmap_0 \uplus \cmdmap,\regstore} \cmdstep{\tid}{\lab} \tup{ \cmdmap_0 \uplus \cmdmap',\regstore}
}
\end{mathpar}}

%\vspace{-5pt}
\caption{Small-step semantics of command pools  ($\tup{\cmdmap,\regstore} \cmdstep{\tid}{\lab_\epsl} \tup{\cmdmap',\regstore'}$)}
\label{fig:cmdpoolsem}
\end{figure}
}

\pcmdsem
\cmdsem
\cmdpoolsem

\paragraph{\textbf{Program semantics.}}
We provide small-step operational semantics to commands
independently of the memory system.
To % be able
connect this semantics to a given memory system,
its steps are instrumented with labels, as defined next.

\begin{definition}
\label{def:label}
A \emph{label} $\lab$ takes one of the following forms:
a read $\rlab{}{\loc}{\valr}$,
a write $\wlab{}{\loc}{\valw}$,
a read-modify-write $\ulab{}{\loc}{\valr}{\valw}$,
a fork $\forklab{\tid_1}{\tid_2}$,
or a join $\joinlab{\tid_1}{\tid_2}$,
where $\loc \in \Loc$, $\valr,\valw\in \Val$, and $\tid_1,\tid_2\in\Tid$.
We denote by $\Lab$ the set of all labels. 
%The functions $\lTYP$, $\lLOC$, $\lVALR$, $\lVALW$, respectively retrieve (when applicable) 
%the type ($\lR/\lW/\lU/\lFORK/\lJOIN$),
%variable ($\loc$), read value ($\valr$), and written value ($\valw$)
%of a given label.
\end{definition}

\begin{definition}
\label{def:register_store}
A \emph{register store} is a mapping $\regstore : \Reg \to \Val$.
Register stores are extended to expressions as expected. 
We denote by $\Regstores$ the set of all register stores. 
\end{definition}

The semantics of (instrumented) primitive commands is given in \cref{fig:pcmdsem}.
Using this definition, the semantics of commands is given in \cref{fig:cmdsem}.  
Its steps are of the form
$\tup{\cmd,\regstore}\astep{\lab_\epsl} \tup{\cmd',\regstore'}$ where
$\cmd$ and $\cmd'$ are commands, $\regstore$ and $\regstore'$ are
register stores, and $\lab_\epsl\in\Lab \cup \set{\epsl}$ ($\epsl$
denotes a thread internal step). We lift this semantics to
\emph{command pools} as follows. % ,
% as defined next.

\begin{definition}
\label{def:command_pool}
A \emph{command pool} is a non-empty
partial function $\cmdmap$ from thread identifiers to commands,
such that the following hold:
%A command pool
%$\cmdmap$ is \emph{well-formed} iff
\begin{enumerate}[noitemsep,topsep=2pt]
\item $\Tid(\cmdmap(\tid_1)) \cap \Tid(\cmdmap(\tid_2)) = \emptyset$ 
% for every two different thread identifiers
  for every   $\tid_{1} \neq \tid_{2}$ in $\dom{\cmdmap}$.
\item $\tid \nin \Tid(\cmdmap(\tid))$ for every $\tid\in\dom{\cmdmap}$.
\end{enumerate}
%Below we assume that all command pools are well-formed.
 We
 write command pools as sets of the form
 $\set{\tid_1\mapsto\cmd_1 \til \tid_n\mapsto\cmd_n}$.  
\end{definition}

%\begin{definition}
%\label{def:command_transition_label}
%A \emph{command transition label} %$\plab$ 
%is a pair of the form $\tidlab{\tid}{\lab_\epsl}$,
%where $\tid\in\Tid$ and $\lab_\epsl\in\Lab \cup \set{\epsl}$.
%%We denote by $\PLab$ the set of all command transition labels. 
%%The functions $\lTID$ and $\lLAB$
%%respectively retrieve
%%the thread identifier ($\tid$)
%%and the label (or $\epsl$) ($\lab_\epsl$)
%%of a given command transition label.
%%In addition, functions on labels ($\lTYP$, $\lLOC$, ...) 
%%are lifted to command's transition labels in the obvious way.
%\end{definition}

%\begin{definition}
%\label{def:command_state}
%A \emph{command state} is a pair $\cmdstate=\tup{\cmd,\regstore,\threads}$
%where $\cmd$ is a command, $\regstore$ is a register store,
%and $\threads\suq\Tid$ is set of \emph{active thread identifiers}.
%The set of \emph{initial command states}
%is given by $\tup{\cmd,\lambda \reg \ldotp 0,\set{\ctid{0}}}$
%where $\cmd$ is a command.
%\end{definition}

Steps for command pools are given in \cref{fig:cmdpoolsem}.
They take the form 
$\tup{\cmdmap,\regstore}\cmdstep{\tid}{\lab_\epsl} \tup{\cmdmap',\regstore'}$,
where $\cmdmap$ and $\cmdmap'$ are command pools,
$\regstore$ and $\regstore'$ are register stores,
and $\tidlab{\tid}{\lab_\epsl}$
(with $\tid\in\Tid$ and $\lab_\epsl\in\Lab \cup \set{\epsl}$)
 is a \emph{command transition label}.

%%% Local Variables:
%%% mode: latex
%%% TeX-master: "main"
%%% End:

%% file: memory.tex
%!TEX root = main.tex

%\section{Memory Systems}
\paragraph{\textbf{Memory semantics.}}
To give semantics to programs % running
under a % particular
memory model, we synchronize the transitions of a command $\cmd$ with
a memory system.  We leave the memory system parametric, and 
assume that it is represented by a labeled transition system (LTS) $\M$
with set of states denoted by $\M.\lQ$,
%initial states $\M.\linit$,
and steps denoted by $\asteplab{}{\M}$.
The transition labels of  general memory system $\M$ consist of non-silent program
transition labels (elements of $\Tid \times \Lab$) % as well as
and a (disjoint) set $\M.\lTheta$ of internal memory actions,
which is again left parametric (used, \eg for memory-internal propagation of values).

%\begin{definition}%[Labeled transition systems]
%A \emph{labeled transition system} (LTS) is a tuple
%$\A=\tup{\Sigma,Q,Q_0,T}$,
%where
%$\Sigma$ is a set of \emph{transition labels},
%$Q$ is a set of \emph{states},
%$Q_0 \suq Q$ is the set of \emph{initial states},
%and $T \suq  Q \times \Sigma \times Q$ is a set of \emph{transitions}.
%We denote by $\A.\lSigma$, $\A.\lQ$, $\A.\linit$, and $\A.\lT$ the components of an LTS $\A$.
%We write $q \astep{\sigma} q'$ to denote a transition $\tup{q,\sigma,q'}$,
%$\asteplab{q}{\A}$ for the relation 
%$\set{\tup{q,q'} \st q\astep{\sigma} q' \in \A.\lT}$, 
%and $\asteplab{}{\A}$ for $\bigcup_{\sigma\in\Sigma} \asteplab{\sigma}{\A}$. 
%For a sequence $\tr \in \A.\lSigma^*$, we write $\asteplab{\tr}{\A}$ for the composition
%$\asteplab{\tr(1)}{\A} \seq \ldots \seq \asteplab{\tr(\size{\tr})}{\A}$.
%A~sequence $\tr \in \A.\lSigma^*$ such that 
%$q_0 \asteplab{\tr}{\A} q$ for some $q_0\in \A.\linit$ and $q\in\A.\lQ$ is called a \emph{trace} of $\A$. % (or an \emph{$\A$-trace}).
%We denote by $\traces{\A}$ the set of all traces of $\A$. 
%A state $q \in \A.\lQ$ is called \emph{reachable} in $\A$ if $q_0 \asteplab{\tr}{\A} q$ for some $q_0\in \A.\linit$ and $\tr \in \traces{\A}$.
%%Finally, for a trace $\tr$ and a set $\Theta\suq\labelType$ of transition labels,
%%we write $\tr\rst\Theta$ for the longest subsequence of $\tr$ over $\Theta$.
%\end{definition}

\begin{example}
\label{example:SC}
The simple memory system that guarantees sequential consistency is denoted here by \SC. 
This memory system tracks the most recent value
written to each variable and has no internal transitions
($\SC.\lTheta=\emptyset$).  Formally, it is defined by
$\SC.\lQ\defeq \Loc \to \Val$
%$\SC.\linit \defeq \lambda \loc \ldotp 0$ 
and $\asteplab{}{\SC}$ is
given by:
\begin{mathpar}
\small
\inferrule*{
\lab=\rlab{}{\loc}{\valr} \\\\
\msc(\loc)=\valr \\
}{{\msc} \asteptidlab{\tid}{\lab}{\SC} {\msc} }
\hfill
\inferrule*{
\lab=\wlab{}{\loc}{\valw} \\\\
\msc'=\msc[\loc \mapsto \valw]  
}{{\msc} \asteptidlab{\tid}{\lab}{\SC} {\msc'} }
\hfill
\inferrule*{
\lab=\ulab{}{\loc}{\valr}{\valw} \\\\
\msc(\loc)=\valr \\\\
\msc'=\msc[\loc \mapsto \valw]  
}{{\msc} \asteptidlab{\tid}{\lab}{\SC} {\msc'} }
\hfill
\inferrule*{
\lab \in\set{\forklab{\_}{\_},\joinlab{\_}{\_}}
%\lTYP(\lab)\in \set{\lFORK,\lJOIN} 
}{{\msc} \asteptidlab{\tid}{\lab}{\SC} {\msc} }
\end{mathpar}
\end{example}

The composition of a program with a general memory system is defined next. 

\begin{definition}
\label{def:system}
The \emph{concurrent system} induced by a memory system $\M$, denoted by $\cs{\M}$,
is the LTS whose transition labels are the elements of $(\Tid \times (\Lab \cup \set{\epsl})) \uplus \M.\lTheta$;
states are triples of the form $\tup{\cmdmap,\regstore,\memstate}$ 
where $\cmdmap$ is a command pool, 
$\regstore$ is a register store, and $\memstate \in \M.\lQ$;
%initial states are states of the form $\tup{\cmdmap,\regstore,\memstate_0}$
%with $\memstate_0 \in \M.\linit$; 
and the transitions 
are ``synchronized transitions'' of the program and the memory system,
using labels to decide what to synchronize on, formally given by:
\begin{mathpar}
\small
\inferrule*{
\tup{\cmdmap,\regstore} \cmdstep{\tid}{\lab} \tup{\cmdmap',\regstore'}
\\\\
\lab \in \Lab \\\memstate {\asteptidlab{\tid}{\lab}{\M}} \memstate'
}{\tup{\cmdmap,\regstore,\memstate} \asteptidlab{\tid}{\lab}{\cs{\M}} \tup{\cmdmap',\regstore',\memstate'}}
\hfill
\inferrule*{
\tup{\cmdmap,\regstore} \cmdstep{\tid}{\epsl} \tup{\cmdmap',\regstore'}}
{\tup{\cmdmap,\regstore,\memstate} \asteptidlab{\tid}{\epsl}{\cs{\M}} \tup{\cmdmap',\regstore',\memstate}}
\hfill
\inferrule*{
\mlab \in \M.\lTheta  \\\\
\memstate \asteplab{\mlab}{\M} \memstate' 
}{\tup{\cmdmap,\regstore,\memstate} \asteplab{\mlab}{\cs{\M}} \tup{\cmdmap,\regstore,\memstate'}}
\end{mathpar}
\end{definition}

%% file: og-general.tex
\newcommand{\figrgrules}{
\begin{figure}[t]
\begin{mathpar}
\small
\inferrule[skip]{ }{\set{\tid \mapsto \skipc} \satm (\pre,\{\pre\},\emptyset,\pre)}
\and 
\inferrule[com]{ 
\M \vDash \mhoare{\pre}{\tid \mapsto \ipcmd}{\post}
}  
{\set{\tid \mapsto \ipcmd} \satm (\pre, \{\pre,\post\},\{\gc{\pre}{\tid \mapsto \ipcmd}\}, \post)}
\and 
\inferrule*[left=seq]{
\set{\tid \mapsto \cmd_1} \satm (\pre,\cR_1, \cG_1,\mida) \\
\set{\tid \mapsto \cmd_2} \satm (\mida,\cR_2,\cG_2,\post)}
{\set{\tid \mapsto\cmd_1 \seq \cmd_2} \satm (\pre,\cR_1 \cup \cR_2,\cG_1 \cup \cG_2,\post)}
\and 
\inferrule*[left=if]{
\set{\tid \mapsto \cmd_1} \satm (\pre \cap \sem{\exp},\cR_1,\cG_1,\post) \\
\set{\tid \mapsto \cmd_2} \satm (\pre \setminus \sem{\exp},\cR_2,\cG_2,\post)}
{\set{\tid \mapsto \ite{\exp}{\cmd_1}{\cmd_2}} \satm (\pre,\cR_1 \cup \cR_2 \cup \set{\pre},\cG_1 \cup \cG_2,\post)}
\and
\inferrule*[left=while]{
\pre \setminus \sem{\exp} \subseteq \post \\ 
\set{\tid \mapsto \cmd} \satm (\pre \cap \sem{\exp},\cR,\cG,\pre)}
{\set{\tid \mapsto \while{\exp}{\cmd}} \satm (\pre,\cR\cup \set{\pre,\post},\cG,\post)}
%\and
%\inferrule[consq]{
%\set{\tid \mapsto \cmd} \satm (\pre',\cR', \cG', \post') \\
%\pre \subseteq \pre' \\
%\forall \statesetR'\in\cR', \astate \in \statesetR', \astate'\in \allstates \ldotp (\forall \statesetR \in \cR \ldotp \sigma \in \statesetR \implies \sigma' \in \statesetR) \implies  \astate'\in \statesetR'\\
%%\cR' \subseteq \cR \\
%\post' \subseteq \post\\\forall \gc{\prop'}{\tida \mapsto \ipcmd} \in \cG' \ldotp 
%\exists \prop   \ldotp \prop' \subseteq \prop \land \gc{\prop}{\tida \mapsto \ipcmd} \in \cG}
%{\set{\tid \mapsto \cmd} \satm (\pre,\cR, \cG, \post)}
\and
 \inferrule*[left=par]{
\set{\tid_1 \mapsto\cmd_1} \satm (\pre_1,\cR_1,\cG_1, \post_1) \\
\set{\tid_2 \mapsto\cmd_2} \satm (\pre_2,\cR_2,\cG_2,\post_2) \\\\
\pre \subseteq \pre_1 \cap \pre_2 \\
\post_1 \cap \post_2 \subseteq \post \\
\tup{\cR_1,\cG_1} \text{ and } \tup{\cR_2,\cG_2} \text{ are non-interfering}}
{\set{\tid_1 \mapsto\cmd_1} \uplus \set{\tid_2\mapsto \cmd_2} \satm (\pre,\cR_1 \cup \cR_2 \cup \{\pre,\post\},\cG_1 \cup \cG_2 ,\post)} 
\and
 \inferrule*[left=fork-join]{
\M \vDash \mhoare{\pre}{{\tid}\mapsto{\forklab{\tid_1}{\tid_2}}}{\pre'} \\
\M \vDash \mhoare{\post'}{{\tid}\mapsto{\joinlab{\tid_1}{\tid_2}}}{\post} \\\\
\set{\tid_1 \mapsto\cmd_1} \uplus \set{\tid_2\mapsto \cmd_2} \satm (\pre',\cR,\cG,\post') \\\\
\cG' = \cG \cup \{ \gc{\pre}{\tid \mapsto \forkcmd{\tid_1}{\tid_2}}, \gc{\post'}{\tid \mapsto \joincmd{\tid_1}{\tid_2}}\}
}{\set{\tid \mapsto \cmd_1 \spar{\tid_1}{}{\tid_2} \cmd_2} \satm (\pre,\cR\cup \set{\pre,\post},\cG',\post)} 
\end{mathpar}
%\vspace{-10pt}
\caption{Generic sequential RG proof rules (letting $\sem{\exp} = \set{\tup{\regstore,\memstate} \st \regstore(\exp)=\btrue}$) }
\label{fig:proof-rules} 
\end{figure} 
}
 
\section{Generic Rely-Guarantee Reasoning}
\label{sec:genRG}
 
 In this section we present our generic RG framework.
%Our RG framework generalizes % both OG~\cite{DBLP:journals/acta/OwickiG76} and %Xu, de Roever and He~
%RG~\cite{DBLP:journals/fac/XuRH97} reasoning. % on \SC. % to arbitrary memory models. 
Rather than committing to a specific assertion language,
our reasoning principles apply on the \emph{semantic level}, using sets of states instead of syntactic assertions. %  and their changes by program commands. 
The structure of proofs still follows program structure, thereby retaining RG's compositionality. % of RG for parallel composition. 
%Since different memory models employ different notions of states, 
%we do not commit to a specific syntactic formalism to represent sets of states
%(like predicate logic). 
%Instead, the framework   reasons on the \emph{semantic} level, 
%i.e.~sets of states and their changes by program commands.  
%where pre- and postconditions and relies are sets of states and guarantees are 
By doing so, we decouple the semantic insights of RG reasoning from a concrete syntax.
Next, we present proof  rules serving as blueprints for memory model specific proof systems. 
An instantiation of this blueprint requires lifting the semantic principles to syntactic ones. 
More specifically, it requires 
\begin{enumerate}
  \item a language with (a) concrete assertions for specifying sets of states and 
  (b) operators that match operations on sets of states (like $\land$ matches $\cap$); and
  \item sound Hoare triples for primitive commands. 
\end{enumerate} 
Thus, each instance of the framework (for a specific memory system)
is left with the task of identifying useful abstractions on states, 
as well as a suitable formalism, for making the generic semantic framework into a proof system. 

\paragraph{\textbf{RG judgments.}}
%Concretely, for the rest of this section, we fix a memory system $\M$.
We let $\M$ be an arbitrary memory system and  $\allstates_\M \defeq \Regstores \times \M.\lQ$. 
%namely tuples consisting of the components of states of a concurrent system that uses $\M$
%except for the command pool $\cmdmap$. 
Properties of programs $\cmdmap$ are stated via \emph{RG judgments}: %  of the form: 
\[ \cmdmap \satm ( \pre,\cR,\cG, \post) \]
where $\pre,\post \subseteq \allstates_\M$,  $\cR \subseteq \powerset{\allstates_\M}$,
and $\cG$ is a set of {\em guarded commands}, each of which takes the form
$\gc{\propg}{\tid \mapsto \alpha }$,
where $\propg \subseteq \allstates_\M$ and
$\alpha$ is either an  (instrumented) primitive command $\ipcmd$
or a fork/join label
(of the form $\forklab{\tid_1}{\tid_2}$ or $\joinlab{\tid_1}{\tid_2}$).
The latter is needed for considering the effect of forks and joins on the memory state.
%\ori{shouldn't we use instrumented commands?}
%For the purpose of defining guarantees, we introduce two  primitive commands not present in our syntax so far: $\forkcmd{\tid_1}{\tid_2}$ and $\joincmd{\tid_1}{\tid_2}$, letting the executing thread run the fork and join step, respectively, as part of the execution of a parallel composition. 
%Specifications $(\pre,\cR, \cG,\post)$ then consist of pre- and postcondition $\pre, \post \subseteq \allstates_\M$, 
%a {\em rely} $\cR \subseteq \powerset{\allstates_\M}$,
%and a {\em guarantee} $\cG$ being a finite set of guarded commands. 
%\begin{itemize} 
%\item pre, post: $\pre, \post \subseteq \Sigma$: set of states \\
%\item relies: $\cR \subseteq 2^\Sigma$ set of sets of states (thread relying on these sets to be stable) \\
%\item guarantees: $\cG$ is a finite set of guarded commands  
%\end{itemize} 

\paragraph{\textbf{Interpretation of RG judgments.}}
%We start by defining the meaning of RG judgements. 
RG judgments $\cmdmap \satm ( \pre,\cR,\cG, \post)$ state that
a terminating run of  $\cmdmap$ starting from a state in $\pre$,
under any concurrent context whose transitions preserve each of the sets of states in $\cR$,
will end in a state in $\post$ and perform only transitions contained in $\cG$. 
%The set $\cG$ consists of {\em guarded commands} in the style of~\cite{DBLP:conf/icalp/LahavV15}
%and specifies the guarantees of a thread. 
%We note that compared to standard RG tuples (as in \cite{DBLP:journals/fac/XuRH97}),
%our judgments include thread identifiers (as a part of the command pool),
%which are technically needed since 
%the effect of memory operation on the memory state may depend on the identity of the
%thread that performs the operation.
To formally define this statement, following the standard model for RG, these judgments are interpreted on {\em computations} of programs. 
Computations arise from runs of the concurrent system  (see \cref{def:system})  
by abstracting away from concrete transition labels
and including arbitrary ``environment transitions'' representing steps of the concurrent context.
%The transitions in a computation are divided into three sorts:
We have: 
\begin{itemize}
\item \emph{Component} transitions of the form
$\tup{\cmdmap,\regstore,\memstate} \trans \compstep \tup{\cmdmap',\regstore',\memstate'}$. 
\item \emph{Memory} transitions, which correspond to internal memory steps 
(labeled with $\mlab \in \M.\lTheta$), of the form
$\tup{\cmdmap,\regstore,\memstate} \trans \memstep \tup{\cmdmap,\regstore,\memstate'}$. 
\item \emph{Environment} transitions of the form
$\tup{\cmdmap,\regstore,\memstate} \trans \envstep \tup{\cmdmap,\regstore',\memstate'}$.
\end{itemize}
Note that memory transitions do not occur in the classical RG presentation
(since \SC does not have internal memory actions).

% \noindent 
A \emph{computation} % $\comp$
is a (potentially infinite) sequence
$$\comp = \tup{\cmdmap_0,\regstore_0,\memstate_0} \trans {a_1}
\tup{\cmdmap_1,\regstore_1,\memstate_1} \trans {a_2} \ldots $$ with
$a_i \in \set{\compstep, \envstep, \memstep}$. We let 
$\tup{\cmdmap_{\last(\comp)},\regstore_{\last(\comp)},\memstate_{\last(\comp)}}$
denotes its last element, when $\comp$ is finite. 
%We let $\Comp$ be the set of all computations.   
%$\tup{\cmdmap_0,\regstore_0,\memstate_0} \trans {a_1} \tup{\cmdmap_1,\regstore_1,\memstate_1} \trans {a_2} \ldots $
We say that $\comp$  is a computation \emph{of a command pool $\cmdmap$} when $\cmdmap_0=\cmdmap$ and for every $i\geq 0$:
\begin{itemize}
%\item $\cmdmap_0=\cmdmap$.
\item If $a_i=\compstep$, then 
$\tup{\cmdmap_i,\regstore_i,\memstate_i} 
%\asteptidlab{\tid}{\lab_\epsl}{\cs{\M}} 
\asteptidlab{\tid}{\lab_\epsl}{\cs{\M}} 
\tup{\cmdmap_{i+1},\regstore_{i+1},\memstate_{i+1}}$
%for some $\tidlab{\tid}{\lab_\epsl}\in\PLab$.  
for some $\tid\in\Tid$ and $\lab_\epsl \in\Lab \cup \set{\epsl}$.  
\item If $a_i=\memstep$, then 
$\tup{\cmdmap_i,\regstore_i,\memstate_i} 
\asteplab{\mlab}{\cs{\M}} 
\tup{\cmdmap_{i+1},\regstore_{i+1},\memstate_{i+1}}$
for some $\mlab\in \M.\lTheta$.  
\end{itemize}
We denote by $\Comp(\cmdmap)$ the set of all computations of a command pool $\cmdmap$.

%\todo{define the labels of a command $\pcmd$, $\labels(\pcmd)$, somewhere, including fork and join} 

To define validity of RG judgments, we use the following definition.

\begin{definition}
Let $\comp=\tup{\cmdmap_0,\regstore_0,\memstate_0} \trans {a_1} \tup{\cmdmap_1,\regstore_1,\memstate_1} \trans {a_2} \ldots$
be a computation,
and $\cmdmap \satm ( \pre,\cR,\cG, \post)$ an RG-judgment.
\begin{itemize}
\item $\comp$ admits $\pre$ if $\tup{\regstore_0,\memstate_0}\in \pre$.
\item $\comp$ admits $\cR$ if 
$\tup{\regstore_i,\memstate_i} \in \statesetR \implies \tup{\regstore_{i+1},\memstate_{i+1}}\in \statesetR$
for every $\statesetR \in \cR$ and $i\geq 0$ with $a_{i+1}=\envstep$.
\item $\comp$ admits $\cG$ if for every $i\geq 0$ with $a_{i+1}=\compstep$ and $\tup{\regstore_i,\memstate_i} \neq \tup{\regstore_{i+1},\memstate_{i+1}}$ there exists $\gc{\prop}{\tid \mapsto \alpha} \in \cG$ such that $\tup{\regstore_i,\memstate_i} \in \prop$ and 
\begin{itemize}
\item if $\alpha=\ipcmd$ is an instrumented primitive command, then
  for some $\lab_\epsl\in\Lab \cup \set{\epsl}$, we have 
  $\tup{\set{\tid\mapsto \ipcmd},\regstore_i, \memstate_i}
  \asteptidlab{\tid}{\lab_\epsl}{\cs{\M}} \tup{\set{\tid\mapsto
      \skipc},\regstore_{i+1},\memstate_{i+1}}$
 
\item if $\alpha \in \set{\forklab{\tid_1}{\tid_2},\joinlab{\tid_1}{\tid_2}}$,
then $\memstate_i {\asteptidlab{\tid}{\alpha}{\M}} \memstate_{i+1}$ and $\regstore_i=\regstore_{i+1}$. 
\end{itemize} 
\item $\comp$ admits $\post$ if 
$\tup{\regstore_{\last(\comp)},\memstate_{\last(\comp)}}\in\post$
whenever $\comp$ is finite and 
$\cmdmap_{\last(\comp)}(\tid) = \skipc$ for every $\tid\in\dom{\cmdmap_{\last(\comp)}}$.
\end{itemize}
We denote by $\assume(\pre,\cR)$ the set of all computations that admit $\pre$ and $\cR$,
and by $\commit(\cG,\post)$ the set of all computations that admit $\cG$ and $\post$.
\end{definition}

Then, \emph{validity} of a judgment if defined as
\[ \models \cmdmap \satm (\pre,\cR,\cG, \post) \defiff \Comp(\cmdmap) \cap \assume(\pre,\cR) \subseteq \commit(\cG,\post)\] 

\figrgrules
 
\paragraph{\textbf{Memory triples.}}
Our proof rules build on {\em memory triples},
which specify pre- and postconditions for primitive commands
for a memory system $\M$. 

\begin{definition}
A \emph{memory triple for a memory system $\M$} is a tuple of the form $\mhoare{\pre}{\tid \mapsto \alpha}{\post}$,
where $\pre,\post\subseteq \allstates_\M$, $\tid\in\Tid$,
and $\alpha$ is either an instrumented primitive command, a fork label, or a join label.  
A memory triple for $\M$ is \emph{valid}, denoted by $\M \vDash \mhoare{\pre}{\tid \mapsto \alpha}{\post}$,
if the following hold for every $\tup{\regstore,\memstate}\in \pre$,
$\regstore'\in\Regstores$ and $\memstate'\in\M.\lQ$:
\begin{itemize}
\item if $\alpha$ is an instrumented primitive command
and $\tup{\set{\tid\mapsto \alpha},\regstore, \memstate} \asteptidlab{\tid}{\lab_\epsl}{\cs{\M}}  \tup{\set{\tid\mapsto \skipc},\regstore',\memstate'}$
for some $\lab_\epsl\in\Lab \cup \set{\epsl}$,
then $\tup{\regstore',\memstate'}\in \post$.
\item If $\alpha \in \set{\forklab{\tid_1}{\tid_2},\joinlab{\tid_1}{\tid_2}}$
and $\memstate {\asteptidlab{\tid}{\alpha}{\M}} \memstate'$,
then $\tup{\regstore,\memstate'}\in \post$.
\end{itemize}
\end{definition}

\begin{example} \label{example:mem-SC} 
For the memory system $\SC$ introduced in \cref{example:SC}, 
we have, \eg memory triples of the form $\SC \vDash \mhoare{\exp\subst{\reg}{\loc}}{\tid \mapsto \readInst{\reg}{\loc}}{\exp}$ 
(where $\exp\subst{\reg}{\loc}$ is the expression $\exp$ with all occurrences of $\reg$ replaced by $\loc$). 
\end{example}

\paragraph{\textbf{RG proof rules.}}
We aim at  proof rules deriving valid RG judgments.
%Our RG proof calculus is independent of a concrete memory system and can be individually instantiated once a memory system is fixed. 
%Once we have derived valid memory triples, we can employ our generic  proof calculus 
%for deriving RG judgements for command pools. 
%Next we introduce our general proof rules for judgments. 
%We write $\vdash \cmdmap \satm (\pre,\cR,\cG, \post)$ to mean provability of the judgment within the proof calculus. 
%
%First of all, for staying on the semantic level in the proof rules, we employ the strongest postcondition of program commands \wrt states (defined in Figure~\ref{fig:sp}) and write the Hoare triple~\cite{DBLP:journals/cacm/Hoare69} $\{ \prop \} \ \tid \mapsto \pcmd \ \{ \propb \}$\footnote{We elide the brackets around command maps here.} as a shorthand for $\forall \astate \in \prop \ldotp \strongpost_\M(\tid, \pcmd)(\astate) \suq \propb$. 
\Cref{fig:proof-rules} lists (semantic) proof rules based on externally provided memory triples.  These rules
basically follows RG reasoning for sequential consistency. For example, rule {\sc seq}
states that RG judgments of commands $\cmd_1$ and $\cmd_2$ can be
combined when the postcondition of $\cmd_1$ and the precondition of
$\cmd_2$ agree, thereby uniting their relies and guarantees. Rule {\sc
  com} builds on memory triples.  The rule {\sc par} for parallel composition
combines judgments for two components when their relies and guarantees
are \emph{non-interfering}.  Intuitively speaking, this means that
each of the assertions that each thread relied on for establishing its
proof is preserved when applying any of the assignments collected in
the guarantee set of the other thread. An example of non-interfering
rely-guarantee pairs is given in step \ref{non-interfering} in
\cref{sec:motivation}.  Formally, non-interference is defined as
follows:
 
\begin{definition} 
\label{def:non-interfering}
Two rely-guarantee pairs $\tup{\cR_1,\cG_1}$ and $\tup{\cR_2,\cG_2}$
are {\em non-interfering} if
$\M \vDash \mhoare{\statesetR \cap \prop}{\tid \mapsto
  \alpha}{\statesetR}$ holds for every $\statesetR \in \cR_1$ and
$\gc{\prop}{\tid \mapsto \alpha} \in \cG_2$, and similarly for every
$\statesetR \in \cR_2$ and
$\gc{\prop}{\tid \mapsto \alpha} \in \cG_1$.
\end{definition} 

In turn, {\sc fork-join} combines the proof of a parallel composition
 with proofs of fork and join steps (which may also affect the memory state). 
Note that the guarantees also involve guarded commands with $\lFORK$ and $\lJOIN$ labels.
%The fork-join rule basically states that whenever the fork step is bringing $\pre$ to $\pre'$, join from $\post'$ to $\post$ and the 
%two threads $\tid_1$ and $\tid_2$ jointly establish the judgement $(\pre',\cR,\cG,\post')$, then their parallel composition brings us from $\pre$ to $\post$ when relying on $\cR$ plus preservation of $\pre$ and $\post$ and guaranteeing $\cG$ (plus the fork step in $\pre$ and the join step in $\post'$). 

Additional rules for consequence and introduction of auxiliary
variables are elided here (they are similar to their \SC
counterparts), and provided in the appendix.
%reasoning about concurrent programs often requires the introduction of auxiliary variables. %(as first remarked in~\cite{DBLP:journals/acta/Clint73}). 
%Our proof rule for auxiliary variables (see appendix) is similar to its \SC counterpart, requiring \eg independence of relies and guarantees w.r.t.~auxiliary variables.  

\paragraph{\textbf{Soundness.}}

To establish soundness of the above system we need an additional requirement
regarding the internal memory transitions
(for \SC this closure vacuously holds as there are no such transitions).
We require all relies in $\cR$ to be \emph{stable under internal memory transitions}, i.e.~for $\statesetR \in \cR$ we require 
\begin{align}
\forall \regstore,\memstate,\memstate', \mlab \in \M.\lTheta \ldotp 
\memstate \asteplab{\mlab}{\M} \memstate' 
\implies (\tup{\regstore,\memstate} \in \statesetR \implies \tup{\regstore,\memstate'} \in \statesetR) \label{int-rely} \tag{{\sf mem}}
\end{align} 
This condition is needed since the memory system can non-deterministically take 
its internal steps, and the component's proof has to be stable under such steps.

With this requirement, we are able to establish soundness.  The proof,
which generally follows~\cite{DBLP:journals/fac/XuRH97} is is given in
the appendix.  We write $\vdash \cmdmap \satm (\pre,\cR,\cG, \post)$
for provability of a judgment using the semantic rules presented
above.

\begin{restatable}[Soundness]{theorem}{RGsoundness} \label{thm:rg-sound}
%  If $\cmd \satm (\pre,\cR,\cG,\post)$ is an RG judgement, then
%   \[ \vdash \cmd \satm (\pre,\cR,\cG,\post) \text{\quad  implies\quad } \vDash \cmd \satm (\pre,\cR,\cG,\post) \]
\!\!\!${}\vdash \cmd \satm (\pre,\cR,\cG,\post) \; \Longrightarrow \; {}\vDash \cmd \satm (\pre,\cR,\cG,\post)$.
\end{restatable}

%%% Local Variables:
%%% mode: latex
%%% TeX-master: "main"
%%% End:

%% file: potential2.tex
%!TEX root = main.tex

\section{Potential-based Memory System for SRA}
\label{sec:potential}

In this section we present the potential-based semantics for Strong Release-Acquire (\SRA),
for which we develop a novel RG logic.
Our semantics is based on the one in 
\cite{DBLP:conf/pldi/LahavB20,DBLP:journals/toplas/LahavB22},
with certain adaptations to make it better suited for Hoare-style reasoning
(see \cref{rem:sra}).
%In particular, our potentials employ stores (mapping of all shared variables to values)
%rather than read options for particular locations as in  
%\cite{DBLP:conf/pldi/LahavB20,DBLP:journals/toplas/LahavB22}.
%We establish the equivalence of our semantics to the one
%in \cite{DBLP:conf/pldi/LahavB20,DBLP:journals/toplas/LahavB22} 
%in Coq (proof scripts to be submitted as artifact). 

In weak memory models, threads typically have different views of the shared memory. 
In \SRA, we refer to a memory snapshot that a thread may observe as a \emph{potential store}:

%In \SC, a thread would at some point in time just have one such store to read from. 
%In \SRA, there is a set of {\em lists} of stores detailing all read options now and in the  {\em future}. 
 
\begin{definition}
\label{def:pot-stores} 
A \emph{potential store} is a function $\potstore : \Loc \to \Val \times \set{\lR,\lU} \times \Tid$. 
We write $\lVALO(\potstore(\loc))$, 
$\lUFLAG(\potstore(\loc))$,
and $\lTID(\potstore(\loc))$ to retrieve the different components of $\potstore(\loc)$.
We denote by $\potstores$ the set of all potential stores.
\end{definition}

Having $\potstore(\loc)=\tup{\val,\lR,\tid}$ allows to read the value $\val$ from $\loc$
(and further ascribes that this read reads from a write performed by thread $\tid$, which is technically
needed to properly characterize the SRA model).
In turn, having $\potstore(\loc)=\tup{\val,\lU,\tid}$ further allows 
to perform an RMW instruction that atomically reads and modifies $\loc$.

Potential stores are collected in \emph{potential store lists} 
describing the values which can (potentially) be read and in what order.

 \begin{cnotation}
 {\rm
Lists over an alphabet $A$ are written as  $L = a_1 \cdottil a_n$
where $a_1 \til a_n\in A$.
 We also use $\cdot$ to concatenate lists,
 and write $L[i]$ for the $i$'th element of $L$
and $\size{L}$ for the length of $L$.  
}\end{cnotation}

A \emph{(potential) store list} is a finite sequence of potential
stores ascribing a possible sequence of stores that a thread can
observe, in the order it will observe them.  
The RMW-flags in these lists have to satisfy certain conditions: 
once the flag for a location is set, it remains set in the rest of the list;
and the flag must be set at the end of the list.
Formally, store lists are defined as follows.

% % certain
% % well-formedness
% conditions, as defined next.
\begin{definition}\label{def:pot-store-list} 
  A \emph{store list} $\L \in \LL$ is a non-empty finite sequence of
  potential stores with \emph{monotone RMW-flags} ending with an
  $\lU$, that is: for all $\loc \in \Loc$,
  \begin{enumerate}
  \item   if
    $\lUFLAG(\L[i](\loc))=\lU$, then $\lUFLAG(\L[j](\loc))=\lU$ for every
    $i < j \leq \size{L}$, and 
  \item $\lUFLAG(\L[\size{\L}](\loc))=\lU$.
  \end{enumerate}
%\begin{itemize}
%   \item $\L$ is non-empty and 
%   \item $\L$ is \emph{monotone in the read-modify-write flags}: $\forall \loc \in \Loc \ \exists i \in \set{1 \til n}$: \\
%         $\exists \tid \in \Tid \; \forall j \geq i. \; \L[j](\loc)=\tup{\cdot,\lU,\tid} \wedge  
%          \forall j < i. \; L[j](\loc) = \tup{\cdot,\lR,\cdot}$.  
%\end{itemize}  
\end{definition}

Now, SRA states ($\SRA.\lQ$) consist of \emph{potential mappings} that assign potentials to threads as defined next. 

\begin{definition}\label{def:sra-state} 
A \emph{potential} $\D$ is a non-empty set of potential store lists.
A \emph{potential mapping} is a function $\DD : \Tid \pfn \powerset{\LL} \setminus \set{\emptyset}$ 
that maps thread identifiers to potentials such that all lists agree on the very final potential store
(that is: $\L_1[\size{L_1}]=  \L_2[\size{L_2}]$ whenever $\L_1 \in \DD(\tid_1)$ and $\L_2 \in \DD(\tid_2)$).
\end{definition}

%ORI: I removed the mentions of initial state from the whole paper
%$\SRA.\linit$   consists of all mappings $\DD$ assigning to $\ctid{0}$ a finite set of potentials with lists of   potential stores $\potstore$ satisfying $\potstore(\loc) = \tup{0,\uflag,\ctid{0}}$ for all $\loc\in \Loc$.  
%\Cref{fig:SRA''-semantics} formulates the steps of SRA.
%Write steps modify the potentials. 
%We write $\potstore[\loc \mapsto \tup{\val,\uflag,\tid}]$ for 
%the potential store obtained by $\potstore$ by 
%updating its $\loc$ entry to $\tup{\val,\uflag,\tid}$,
%and, similarly, $\potstore[\loc \mapsto \uflag]$ for the 
%the potential store obtained by $\potstore$ by 
%updating the RMW-flag of its $\loc$ entry to $\uflag$.
%This notations are lifted pointwise to store lists $\L$.
%\ori{we need to explain the write step here...}
%\DeclareRobustCommand{\brkstore}{\genfrac[]{0pt}{}}
These potential mappings are ``lossy'' meaning that potential stores can be arbitrarily dropped. 
%Read steps read from the  first store of a list. 
%Store lists are ``lossy'' meaning that stores can be arbitrarily dropped. 
In particular, dropping the first store in a list enables reading from the second. 
This is formally done by transitioning from a state $\DD$ to a ``smaller'' state $\DD'$ as defined next.
%This dropping is defined by an ordering on store lists (which extends to potentials
%and  potential functions).

\begin{definition}
\label{def:sqsuq}
The (overloaded) partial order $\sqsuq$ is defined as follows: 
\begin{enumerate}
\item on potential store lists: $\L' \sqsuq \L$ if $\L'$ is a 
nonempty subsequence of $\L$;  
\item on potentials: $\D' \sqsuq \D$ if 
$\forall \L'\in \D'.\; \exists \L\in\D.\; \L' \sqsuq \L$;  
\item on potential mappings:
$\DD' \sqsuq \DD$ if 
$\DD'(\tid)\sqsuq \DD(\tid)$
for every $\tid\in\dom{\DD}$.
\end{enumerate}
\end{definition}

\newcommand{\lz}[1]{{\color{green!60!black}{#1}}}
\newcommand{\lo}[1]{{\color{blue!60!black}{#1}}}

\begin{figure}[t] 
\begin{mathpar}
\small
\inferrule[write]{
\forall \L'\in\DD'(\tid). \; \exists \L \in\DD(\tid) \ldotp \L' = \L[\loc \mapsto \tup{\valw,\lU,\tid}] \\\\
\inarrT{
\forall \tida\in\dom{\DD} \setminus \set{\tid}, \L'\in\DD'(\tida) \ldotp \exists \lz{\L_0},  \lo{\L_1}  \ldotp \\
\qquad \lz{\L_0}  \cdot \lo{\L_1}  \in \DD(\tida) \land   \lo{\L_1}  \in \DD(\tid) \land {} \\
\qquad  \L' = \lz{\L_0}[\loc \mapsto \lR]   \cdot \lo{\L_1}[\loc \mapsto \tup{\valw,\lU,\tid}]}
%\qquad  \forall \delta \in \D_0 \ldotp \lUFLAG(\delta(\loc))=\lR  \\
}{{\DD} {\asteptidlab{\tid}{\wlab{}{\loc}{\valw}}{\SRA}}{\DD'}}
\and\inferrule[lose]{
\DD' \sqsuq \DD
}{\DD  \astep{\varepsilon}_{\SRA} {\DD'}}
\and
\inferrule[dup]{\DD \preceq \DD'}{\DD \astep{\varepsilon}_{\SRA} {\DD'}} 
\and
\inferrule[read]{
  \exists \tida \ldotp \forall \L \in \DD(\tid) \ldotp \inarrT{\lVALO(\L[1](\loc)) = \valr \land {} \\
    \lTID(\L[1](\loc)) = \tida}
}{{\DD} {\asteptidlab{\tid}{\rlab{}{\loc}{\valr}}{\SRA}}{\DD}}
\and
\inferrule[rmw]{
 \forall \L \in \DD(\tid) \ldotp \lUFLAG(\L[1](\loc)) = \lU \\\\
\DD {\asteptidlab{\tid}{\rlab{}{\loc}{\valr}}{\SRA}} {\DD} \\
\DD {\asteptidlab{\tid}{\wlab{}{\loc}{\valw}}{\SRA}} {\DD'}
}{{\DD} {\asteptidlab{\tid}{\ulab{}{\loc}{\valr}{\valw}}{\SRA}} {\DD'}}
\and
\inferrule[fork]{
\DD_{\text{new}} = \set{\tid_1 \mapsto \DD(\tid), \tid_2 \mapsto \DD(\tid)} \\\\
\DD' = \DD \rst{\dom{\DD}\setminus\set{\tid}} \uplus \DD_{\text{new}}
}{{\DD} {\asteptidlab{\tid}{\forklab{\tid_1}{\tid_2}}{\SRA}} \DD'}
\and 
\inferrule[join]{
\DD_{\text{new}} = \set{\tid \mapsto \DD(\tid_1) \cap \DD(\tid_2)} \\\\
\DD' = \DD \rst{\dom{\DD}\setminus\set{\tid_1,\tid_2}} \uplus \DD_{\text{new}}
}{{\DD} {\asteptidlab{\tid}{\joinlab{\tid_1}{\tid_2}}{\SRA}} \DD'}
\end{mathpar}
%\vspace{-10pt}
\caption{Steps of $\SRA$ (defining $\potstore[\loc \mapsto \tup{\val,\uflag,\tid}](\loca) = \tup{\val,\uflag,\tid}$ if $\loca = \loc$ and $\potstore(\loca)$ else,
and $\potstore[\loc \mapsto \lR]$ to set all RMW-flags for $\loc$ to $\lR$; both pointwise lifted to lists)}
\label{fig:SRA''-semantics} 
\end{figure}

We also define $\L \preceq \L'$  if $\L'$ is obtained from $\L$ by duplication of some stores (\eg 
$\potstore_1 \cdot \potstore_2 \cdot \potstore_3 \preceq \potstore_1 \cdot \potstore_2 \cdot \potstore_2 \cdot \potstore_3$). 
%This duplication is reflected by the (again memory internal) step \textsc{dup} of Fig.~\ref{fig:SRA''-semantics}. 
This is lifted to potential mappings as expected.

\Cref{fig:SRA''-semantics} defines the transitions of $\SRA$. 
The {\sc lose} and {\sc dup} steps account for losing and duplication in potentials. 
Note that these are both internal memory transitions (required to preserve relies as of (\ref{int-rely})). 
The {\sc fork} and {\sc join} steps distribute potentials on forked threads and join them at the end. 
The {\sc read} step obtains its value from the first store in the lists of the potential of the reader,
provided that all these lists agree on that value and the writer thread identifier. 
{\sc rmw} steps atomically perform a read and a write step
where the read is restricted to an $\lU$-marked entry.

Most of the complexity is left for the {\sc write} step.
It updates to the new written value for the writer thread $\tid$.
For every other thread, it updates a \emph{suffix} ($\lo{\L_1}$) of the store list with the new value.
For guaranteeing causal consistency this updated suffix cannot be arbitrary:
it has to be in the potential of the writer thread ($\lo{\L_1}  \in \DD(\tid)$).
This is the key to achieving the ``shared-memory causality principle'' of \cite{DBLP:journals/toplas/LahavB22},
which ensures causal consistency.

\newcommand{\up}[1]{{\color{red!70!black}{#1}}}
\begin{example}
Consider again the MP program from \cref{fig:mp-SRA}. 
After the initial fork step, threads $\ctid{1}$ and $\ctid{2}$
may have the following store list in their potentials:  
$$\small \clist{L} =   
\brkstore{\cloc{x}\mapsto \tup{0,\lU,\ctid{0}}}{\cloc{y} \mapsto \tup{0,\lU,\ctid{0}}}\cdot
\brkstore{\cloc{x}\mapsto \tup{0,\lU,\ctid{0}}}{\cloc{y} \mapsto \tup{0,\lU,\ctid{0}}}\cdot
\brkstore{\cloc{x}\mapsto \tup{0,\lU,\ctid{0}}}{\cloc{y} \mapsto \tup{0,\lU,\ctid{0}}}.$$
Then, $\writeInst{\cloc{x}}{1}$ by $\ctid{1}$ can generate the following store list for $\ctid{2}$:
$$\small \clist{L_2} =   
\brkstore{\cloc{x}\mapsto \tup{0,\up{\lR},\ctid{0}}}{\cloc{y} \mapsto \tup{0,\lU,\ctid{0}}}\cdot
\brkstore{\cloc{x}\mapsto \tup{\up{1},\up{\lU},\up{\ctid{1}}}}{\cloc{y} \mapsto \tup{0,\lU,\ctid{0}}}\cdot
\brkstore{\cloc{x}\mapsto \tup{\up{1},\up{\lU},\up{\ctid{1}}}}{\cloc{y} \mapsto \tup{0,\lU,\ctid{0}}}.$$
Thus $\ctid{2}$ keeps the possibility of reading the ``old'' value of $\cloc{x}$. % however, could also drop this option later. 
For $\ctid{1}$ this is different: the model allows the writing thread to only see its new value of $\cloc{x}$
and all entries for $\cloc{x}$ in the store list are updated. 
Thus, for $\ctid{1}$ we obtain store list 
$$\small \clist{L_1} =   
\brkstore{\cloc{x}\mapsto \tup{\up{1},\up{\lU},\up{\ctid{1}}}}{\cloc{y} \mapsto \tup{0,\lU,\ctid{0}}}\cdot
\brkstore{\cloc{x}\mapsto \tup{\up{1},\up{\lU},\up{\ctid{1}}}}{\cloc{y} \mapsto \tup{0,\lU,\ctid{0}}}\cdot
\brkstore{\cloc{x}\mapsto \tup{\up{1},\up{\lU},\up{\ctid{1}}}}{\cloc{y} \mapsto \tup{0,\lU,\ctid{0}}}.$$
%(and by dropping also $\clist{L_3} = \set{\cloc{x}\mapsto \tup{1,\lU,\ctid{1}},\cloc{y} \mapsto \tup{0,\lU,\ctid{0}}}$). 
Next, when $\ctid{1}$ executes $\writeInst{\cloc{y}}{1}$,
again, the value for $\cloc{y}$ has to be updated to $1$ in  $\ctid{1}$ yielding 
$$\small \clist{L'_1} =   
\brkstore{\cloc{x}\mapsto \tup{1,\lU,\ctid{0}}}{\cloc{y} \mapsto \tup{\up{1},\up{\lU},\up{\ctid{1}}}}\cdot
\brkstore{\cloc{x}\mapsto \tup{1,\lU,\ctid{1}}}{\cloc{y} \mapsto \tup{\up{1},\up{\lU},\up{\ctid{1}}}}\cdot
\brkstore{\cloc{x}\mapsto \tup{1,\lU,\ctid{1}}}{\cloc{y} \mapsto \tup{\up{1},\up{\lU},\up{\ctid{1}}}}.$$
For $\ctid{2}$ the write step may change $\clist{L_2}$ to 
$$\small \clist{L_2'} =   
\brkstore{\cloc{x}\mapsto \tup{0,\lR,\ctid{0}}}{\cloc{y} \mapsto \tup{0,\up{\lR},\ctid{0}}}\cdot
\brkstore{\cloc{x}\mapsto \tup{1,\lU,\ctid{1}}}{\cloc{y} \mapsto \tup{0,\up{\lR},\ctid{0}}}\cdot
\brkstore{\cloc{x}\mapsto \tup{1,\lU,\ctid{1}}}{\cloc{y} \mapsto \tup{\up{1},\up{\lU},\up{\ctid{1}}}}.$$
Thus, thread $\ctid{2}$ can still see the old values,
or lose the prefix of its list and see the new values.
Importantly, it cannot read $1$ from $\cloc{y}$ {\em and then} $0$ from $\cloc{x}$. 
Note that $\writeInst{\cloc{y}}{1}$ by $\ctid{1}$ \emph{cannot} modify $\clist{L_2}$ to the list
$$\small \clist{L_2''} =   
\brkstore{\cloc{x}\mapsto \tup{0,\lR,\ctid{0}}}{\cloc{y} \mapsto{\tup{\up{1},\up{\lU},\up{\ctid{1}}}}}\cdot
\brkstore{\cloc{x}\mapsto \tup{1,\lU,\ctid{1}}}{\cloc{y} \mapsto \tup{\up{1},\up{\lU},\up{\ctid{1}}}}\cdot
\brkstore{\cloc{x}\mapsto \tup{1,\lU,\ctid{1}}}{\cloc{y} \mapsto \tup{\up{1},\up{\lU},\up{\ctid{1}}}},$$
as it requires $\ctid{1}$ to have $\clist{L_2}$ in its \emph{own potential}.
This models the intended semantics of message passing under causal consistency.
% no such justification exists as  $\ctid{1}$ has set $\cloc{x}$ to be 1 before setting $\cloc{y}$ to 1. 
%In our example: for $\clist{L_2'}$ the required justification $\L_1$ of $\ctid{1}$  is   $\clist{L_3}$. 
\end{example}

%\begin{definition}
%\label{def:sral}
%The semantics of \SRA is defined as follows. 
%\begin{itemize}
%\item $\SRA.\lQ$ is the set of potential functions $\DD$ assigning a potential to a subset $\dom{\DD}\suq\Tid$ of thread identifiers;
%\item $\SRA.\linit$  
%     consists of all functions $\DD$ assigning to $\ctid{0}$ a finite set of potentials such that all its potential stores $\potstore$ satisfy $\potstore(\loc) = \tup{0,\uflag,\ctid{0}}$ for all $\loc\in \Loc$.  
%\item Transitions are defined by the rules given in Figure~\ref{fig:SRA''-semantics}. 
%\end{itemize}
%\end{definition}

The next theorem establishes the equivalence of \SRA as defined above
and \textsf{opSRA} from \cite{DBLP:journals/toplas/LahavB22},
which is an (operational version of) the standard strong release-acquire declarative semantics~\cite{DBLP:conf/popl/LahavGV16,DBLP:journals/siglog/Lahav19}.
(As a corollary, we obtain the equivalence between the potential-based system from \cite{DBLP:journals/toplas/LahavB22}
and the variant we define in this paper.)

Our notion of equivalence employed in the theorem is {\em trace equivalence}. We let a trace of a memory system be a sequence of transition labels, ignoring $\epsl$ transitions, 
and consider traces of \SRA starting from an initial state $\lambda \tid \in \set{\ctid{1} \til \ctid{N}} \ldotp \set{ \tup{\lambda \loc \ldotp \tup{0,\lU,\ctid{0}}}}$ 
and traces of \textsf{opSRA} starting from the initial execution graph that consists of a write event to every location writing $0$ by a distinguished initialization thread $\ctid{0}$. 

\begin{theorem}
  A trace is generated by \SRA iff it is generated by \textsf{opSRA}. 
\end{theorem} 

%\begin{theorem}
%A trace (\ie a sequence of transition labels, ignoring $\epsl$ transitions) is generated by the \SRA memory system (defined in this section)
%starting from an initial state $\lambda \tid \in \set{\ctid{1} \til \ctid{N}} \ldotp \set{ \tup{\lambda \loc \ldotp \tup{0,\lU,\ctid{0}}}}$
%iff it is generated by  \normalfont\textsf{opSRA} from \cite{DBLP:journals/toplas/LahavB22}
%starting from the initial execution graph that consists of a write event to every location writing $0$ by a distinguished initialization thread $\ctid{0}$.
%\end{theorem}

The proof is of this theorem is by simulation arguments (forward
simulation in one direction and backward for the converse).  It is
mechanized in Coq and available in~\cite{CAV-Artifact}.  The
mechanized proof does not consider fork and join steps, but they can
be straightforwardly added.

% \ori{not sure how convincing this is, but
%   don't want to add more technicalities}

%%% Local Variables:
%%% mode: latex
%%% TeX-master: "main"
%%% End:

%% file: logic2.tex
%!TEX root = main.tex

\newcommand{\lorafig}{
\begin{figure}[t] 
$\begin{array}{@{} l @{\quad }r l l l @{}}
\text{\em extended expressions} & %\EExp \ni 
& \eexp & ::=  & \exp \ALT \loc \ALT \lR(\loc) \ALT \eexp + \eexp \ALT % \eexp = \eexp \ALT
                                                   \neg \eexp \ALT \eexp \land \eexp \ALT % \eexp \lor \eexp \ALT
                                                   \ldots \\ 
\text{\em interval assertions} & % \Inter \ni 
& \inter & ::= & [\eexp] \ALT \inter \chop \inter \ALT \inter \wedge \inter \ALT \inter \lor \inter \\  
\text{\em assertions} & % \Assert \ni 
& \gassert, \postf& ::= & \potassert{\tid}{\inter} \ALT \exp \ALT \gassert \land \gassert \ALT \gassert \lor \gassert
 \end{array}$
% \vspace{-10pt}
\caption{Assertions of \lora} 
\label{fig:logic} 
\end{figure}}

\section{Program Logic} \label{sec:logic}

For the instantiation of our RG framework to \SRA, we next 
(1) introduce the assertions of the logic \lora and (2) specify memory triples for \lora. 
%For the memory system \SRA we need assertions describing properties of thread's potentials, 
%i.e.~properties about {\em lists} of stores and their orderings. 
Our logic is inspired by {\em interval logics} like Moszkowski's ITL~\cite{DBLP:journals/corr/abs-1207-3816} or duration calculus~\cite{DBLP:journals/ipl/ChaochenHR91}. 

\paragraph{\textbf{Syntax and semantics.}}
\Cref{fig:logic} gives the grammar of \lora. 
We base it on \emph{extended expressions} which---besides registers---can also involve locations   as well as 
 expressions of the form $\lR(\loc)$ (to indicate RMW-flag   $\lR$).  
Extended expressions $\eexp$ can hold on entire {\em intervals} of a store list (denoted $[\eexp]$). Store lists can be split into intervals satisfying different interval expressions ($\inter_1 \chop \ldots \chop \inter_n$) using the ``$\chop$'' operator (called ``chop''). 
%(chop often elided). 
In turn, $\potassert{\tid}{\inter}$ means that all store lists in $\tid$'s potential satisfy $\inter$. 
For an assertion $\gassert$, we let $\fv(\gassert) \subseteq \Reg \cup \Loc \cup \Tid$ be the set of registers, locations and thread identifiers occurring in $\gassert$, and write $\lR(\loc) \in \gassert$ to indicate that the term $\lR(\loc)$ occurs in $\gassert$. 
% and $\lTID(\loc)\in\Tid$ (to indicate the writing thread). 

%An \emph{atomic potential assertion} $\inter$ is a list of extended expressions.

As an example consider again  MP (\cref{fig:mp-SRA}). We would like to express that   $\ctid{2}$ upon seeing $\cloc{y}$ to be 1 cannot see the old value 0 of $\cloc{x}$ anymore. In \lora this is expressed as $\potassert{\ctid{2}}{[\cloc{y} \neq 1] \chop [\cloc{x} = 1]}$: the store lists of $\ctid{2}$ can be split into two  intervals (one possibly empty), the first satisfying $\cloc{y} \neq 1$ and the second $\cloc{x}=1$. 
%Each of these intervals can also be empty. 

Formally, an assertion $\gassert$ describes  register stores coupled with \SRA states: 

\lorafig

\begin{definition} 
Let 
$\regstore$ be a register store, 
$\potstore$ a potential store, 
$\L$ a store list, 
and $\DD$ a potential mapping.  
We let $\sem{\exp}_{\tup{\regstore,\potstore}} = \regstore(\exp)$,
$\sem{\loc}_{\tup{\regstore,\potstore}} = \potstore(\loc)$, 
and $\sem{\lR(\loc)}_{\tup{\regstore,\potstore}} = {\sf if}\
\lUFLAG(\potstore(\loc)) = \lR\ {\sf then} \ \btrue\ {\sf else}\ \bfalse$.  
The extension of this notation to any extended expression $\eexp$ is standard.
The validity of assertions in $\tup{\regstore,\DD}$, denoted by $\tup{\regstore,\DD} \models \gassert$, is defined as follows:
\begin{enumerate}[topsep=2pt]
%\item  $\tup{\regstore,\potstore} \models \eexp$ if  $\sem{\eexp}_{\tup{\regstore,\potstore}} = 1$. 
%\item  $\tup{\regstore,\L} \models \eexp$ if $\tup{\regstore,\potstore} \models \eexp$ for all $\potstore \in \L$. 

\item $\tup{\regstore,\L} \models {[\eexp]}$ if $\sem{\eexp}_{\tup{\regstore,\potstore}} = \btrue$ for every $\potstore \in \L$. 

\item $\tup{\regstore,\L} \models {\inter_1 \chop \inter_2}$ if
$\tup{\regstore,\L_1}  \models \inter_1$
and
$\tup{\regstore,\L_2}  \models \inter_2$
for some (possibly empty)
  $\L_1$ and $\L_2$ such that $\L=\L_1 \cdot \L_2$. 

\item $\tup{\regstore,\L} \models {\inter_1 \land \inter_2}$ if 
   $\tup{\regstore,\L} \models \inter_1$ and 
  $\tup{\regstore,\L} \models \inter_2$  (similarly for $\lor$).

 \item $\tup{\regstore,\DD} \models \potassert{\tid}{\inter}$ if 
   $\tup{\regstore,\L} \models \inter$ for every $\L \in \DD(\tid)$. 
% \item $\tup{\regstore,\DD} \models \exp$ if $\regstore(\exp)=1$. 

% \item $\tup{\regstore,\DD} \models \potassert{\tid}{[\eexp_1] \chop \ldots \chop [\eexp_n]}$ if 
%    $\forall \L \in \DD(\tid)\; \exists \L_1 \til \L_n. L=\L_1 \cdottil \L_n\wedge \forall i, 1 \leq i \leq n. \tup{\regstore,\L_i} \models \eexp_i$. \bd{can use the rule for $I; I$?}
\item $\tup{\regstore,\DD} \models \exp$ if $\regstore(\exp)=\btrue$. 
 
         \item $\tup{\regstore, \DD} \models \gassert_1 \land \gassert_2$ if $\tup{\regstore, \DD} \models \gassert_1$ and $\tup{\regstore, \DD} \models \gassert_2$  (similarly for $\lor$). 
         %\item $\tup{\regstore, \DD} \models \gassert_1 \lor \gassert_2$ if $\tup{\regstore, \DD} \models \gassert_1$ or $\tup{\regstore, \DD} \models \gassert_2$. 
  
\end{enumerate}
\end{definition}

 Note that with $\land$ and $\lor$ as well as negation on expressions,\footnote{Negation just occurs on the level of simple expressions $\exp$ which is sufficient for calculating $\pre \setminus \sem{\exp}$ required in rules {\sc if} and {\sc while}.} the logic provides the operators on sets of states necessary for an instantiation of our RG framework.
Further, the requirements from \SRA states guarantee certain properties:
\begin{itemize}[noitemsep,topsep=2pt]
  \item For $\gassert_1 = \potassert{\tid}{[\eexp^\tid_1]\chop \ldots \chop [\eexp^\tid_n]}$ and $\gassert_2 = \potassert{\tida}{[\eexp^\tida_1]\chop \ldots \chop [\eexp^\tida_m]}$: if $\eexp^\tid_i \land \eexp^\tida_j \implies \false$ for all $1 \leq i\leq n$ and $1 \leq j \leq m$, then $\gassert_1 \land \gassert_2 \implies \false$  
  (follows from the fact that all lists in potentials are non-empty and agree on the last  store). 
  \item If $\tup{\regstore, \DD} \models \potassert{\tid}{[\lR(\loc)]\chop [\eexp]}$, then every list $\L \in \DD(\tid)$ contains a non-empty suffix satisfying $\eexp$ (since all lists have to end with RMW-flags set on). 
\end{itemize}

All assertions are preserved by   steps \textsc{lose} and \textsc{dup}. % \ie they all satisfy (\ref{int-rely}). 
This stability  is required by our RG framework (condition (\ref{int-rely}))\footnote{Such stability requirements are also common to other reasoning techniques for weak memory models, \eg~\cite{DBLP:journals/tocl/DohertyDDW22}.}. 
Stability is achieved here because 
%(a) store lists are non-empty (\ie an assertion $[\loc \neq \val]$ cannot become true because of dropping to an empty sequence) 
%\ori{I don't follow the logic here. the problem is with assertions becoming false when drooping}
%and (b) 
negations  occur on the level of (simple) expressions only 
(\eg we cannot have $\neg (\potassert{\tid}{[\loc = \val])}$, meaning that $\tid$ must have
a store in its potential whose value for $\loc$ is not $\val$, which would not be stable under {\sc lose}). 
%As a consequence, there are no extra checks about (\ref{int-rely}) required for relies occuring in concrete examples. 

\begin{proposition} 
If $\tup{\regstore, \DD} \models \gassert$ and 
$\DD  \astep{\varepsilon}_\SRA \DD'$,
%(steps \textsc{lose} or \textsc{dup}),
then $\tup{\regstore, \DD'} \models \gassert$. 
%Suppose that $\DD  \astep{\varepsilon}_\SRA \DD'$ (steps \textsc{lose} or \textsc{dup}). 
%Then, the following holds for every $\regstore \in \Regstores$ and $\gassert \in \Assert$:  
%\[ \tup{\regstore, \DD} \models \gassert \text{ implies } \tup{\regstore, \DD'} \models \gassert \] 
\end{proposition} 

%\hw{prove it?} 

%\noindent This  requires the restriction of potentials to non-empty lists as of \cref{def:pot-stores}.  
%\ori{I don't get this sentence}

%%% Local Variables:
%%% mode: latex
%%% TeX-master: "main"
%%% End:

%% file: axioms.tex
%!TEX root = main.tex

\paragraph{\textbf{Memory triples.}} 
%\label{sec:memory-triples-sra} 
  %\hw{the $\inter_1 \land \inter_2$ is currently not needed, the rule is gone}
Assertions in \lora describe sets of states, thus can be used to
formulate memory triples.  \Cref{fig:mem-triples} gives the base triples
for the different primitive instructions.
 %$\faddInstn$ and $\casInstn$ are elided due to lack of space (and because they do not occur in our examples). 

 \begin{figure}[t]
   \centering
   \begin{small}
     \scalebox{0.95}{
\begin{tabular}[t]{|l|r@{}c@{}l|l@{}|}
  \hline
  \scalebox{0.9}{Assumption} & Pre  & Command  & Post & Reference \\
  \hline
              & $\assert{\gassert\subst{\reg}{\exp}}$ & $\tid \mapsto \assignInst{\reg}{\exp}$
                                & $\assert{\gassert}$ & {\sc Subst-asgn} \\
                                \hline  
%  $\reg \notin \fv(\gassert)$ & $\assert{\gassert}$ & $\tid \mapsto \assignInst{\reg}{\exp}$
 %                               & $\assert{\gassert}$ & {\sc Stable-asgn} \\
  $\loc \notin \fv(\gassert)$ & $\assert{\gassert}$
                     & $\tid \mapsto  \wrtInst{\loc}{\exp}$ & $\assert{\gassert}$ & {\sc Stable-wr} \\
  $\reg \notin \fv(\gassert)$ & $\assert{\gassert}$
                     & $\tid \mapsto \readInst{\reg}{\loc}$ & $\assert{\gassert}$ & {\sc Stable-ld} \\
  $\tid \notin \fv(\gassert)$ & $\assert{\gassert}$ & $\tid \mapsto \forklab{\tid_1}{\tid_2}$ & $\assert{\gassert}$ & {\sc Stable-fork} \\
  $\tid \notin \fv(\gassert)$ & $\assert{\gassert}$
                     & $\tid \mapsto \joinlab{\tid_1}{\tid_2}$ & $\assert{\gassert}$ & {\sc Stable-join}\\
%  $\reg, \loc \notin \fv(\gassert)$
 %             & $\assert{\gassert}$ & $\tid \mapsto \faddInst{\reg}{\loc}{\exp}$ & $\assert{\gassert}$ & {\sc Stable-faad} \\
 % $\reg, \loc \notin \fv(\gassert)$
 %             & $\assert{\gassert}$ & $\tid \mapsto \casInst{\reg}{\loc}{\exp_\lR}{\exp_\lW}$ & $\assert{\gassert}$ & {\sc Stable-cas}\\
%  $\loc \notin \fv(\gassert)$ 
%              & $\assert{\gassert}$ & $\tid \mapsto \swapInst{\loc}{\exp}$ & $\assert{\gassert}$ & {\sc Stable-swap} \\
\hline
  
              & $\assert{\exp  \land \potassert{\tid}{\inter} }$ & $\tid \mapsto \forklab{\tid_1}{\tid_2}$ & $ \assert{\exp \land \potassert{\tid_1}{\inter} \land \potassert{\tid_2}{\inter}  }$ & {\sc Fork}
  \\
  
              & $\assert{\exp \land \potassert{\tid_1}{\inter}  \land  \potassert{\tid_2}{\inter} }$ & $\tid \mapsto \joinlab{\tid_1}{\tid_2}$ & $\assert{\exp \land \potassert{\tid}{\inter}  }$ & {\sc Join}\\
\hline

              & $\assert{\true}$ & $\tid \mapsto \wrtInst{\loc}{\exp}$ & $\assert{\potassert{\tid}{[ \loc=\exp]}}$  & {\sc Wr-own} \\

 % & $\assert{\true}$ & $\tid \mapsto \writeInst{\loc}{\exp}$ & $\assert{\exp \neq \exp' \implies \\
  %\quad \potassert{\tid}{[\natprop{\loc}{\exp'}]}}$ & {\sc Wr-own-2}* \\
  
 % $\tida \neq \tid$             & $\assert{\true}$
 %                     & $\tid \mapsto \wrtInst{\loc}{\exp}$ & $\assert{\potassert{\tida}{[\lR(\loc)] \chop [\loc = \exp]}}$  & {\sc Wr-other-1} \\
    $% \tida \neq \tid, 
  \lR(\loc) \notin \inter$             & $\assert{\potassert{\tida}{\inter } }$
                     & $\tid \mapsto \wrtInst{\loc}{\exp}$ & $\assert{\potassert{\tida}{\inarrT{(\inter \wedge [\lR(\loc)]) \chop [x = e]}}  }$  & {\sc Wr-other-1} \\
  $\inarr{\loc \notin \fv(\inter_\tid), \\
  % \tida \neq \tid, 
  \lR(\loc) \notin \inter}$             & $\assert{\potassert{\tid}{\inter_\tid} \land  \potassert{\tida}{\inter \chop \inter_\tid}}$ 
                    &  $\tid \mapsto \wrtInst{\loc}{\exp}$ & $\assert{\potassert{\tida}{\inter \chop \inter_\tid}}$ & {\sc Wr-other-2} \\
                    $\loc \notin \fv(\inter_\tid)% , \tida \neq \tid
  $  & $\assert{\potassert{\tid}{\inter_\tid}}$ & ${\tid \mapsto \wrtInst{\loc}{\exp}}$ & $\assert{\potassert{\tida}{[\lR(\loc)] \chop \inter_\tid }}$ & {\sc Wr-other-3 }
                    \\
  \hline  
%                    
%  $\inarr{\loc, \reg \notin \fv(\inter_\tid), \\
%  \tida \neq \tid, \\
%  \lR(\loc)\notin \inter, \reg \notin \fv(\inter)}$             & $\assert{\potassert{\tid}{\inter_\tid} \land {} \\ \potassert{\tida}{\inter \chop \inter_\tid}}$ 
%                    &  $\tid \mapsto \rmwInstn(\reg,\loc)$ & $\assert{\potassert{\tida}{\inter \chop \inter_\tid}}$ & {\sc Rmw-other-2} 
%                    \\
%                    \\
%

    %                $\loc \notin \fv(\inter_\tid), \tida \neq \tid$             & $\assert{\potassert{\tid}{\inter_\tid} \land \potassert{\tida}{\inter \chop \inter_\tid}}$ 
     %               &  $\tid \mapsto \writeInst{\loc}{\exp}$ & $\assert{\potassert{\tida}{(\inter \land [\lR(\loc)])  \chop (\inter_\tid \land [\loc = \exp]}}$ & {\sc Wr-other-3} \\
              %& $\assert{\potassert{\tida}{\inter}}$ & $\tid \mapsto \writeInst{\loc}{\exp}$ & $\assert{\exp \neq \exp' \implies \\
  %\quad \potassert{\tida}{\inter \chop [\natprop{\loc}{\exp'}]}}$ & {\sc Wr-other-2}*
%  & $\assert{\true}$ & $\tid \mapsto \swapInstn(\loc,\exp_\lW)$ & $\assert{\potassert{\tid}{[ \loc=\exp_\lW]}}$  & {\sc Swap-own-1} \\
  $\loc \notin \fv(\inter)$ & $\assert{\potassert{\tid}{[\lR(\loc)] \chop  \inter}}$ & $\tid \mapsto \swapInst{\loc}{\exp}$ & $\assert{\potassert{\tid}{ \inter}}$  & {\sc Swap-skip} \\
  \hline
\end{tabular}}
\end{small}
\caption{Memory triples for \lora using
  $\wrtInstn \in \set{\swapInstn,\writeInstn}$ and assuming
  $\tid \neq \tida$}
\label{fig:mem-triples} 
\end{figure}
We see the standard \SC rule of assignment ({\sc Subst-asgn}) for registers followed by 
a number of stability rules detailing when assertions are not affected by instructions. 
Axioms {\sc Fork} and {\sc Join} describe the transfer of properties from forking thread to forked threads and back. 

The next four axioms in the table concern write instructions (either $\swapInstn$ or $\writeInstn$). 
They reflect the semantics of writing in \SRA: 
 (1) In the writer thread $\tid$ all stores in all lists get updated (axiom {\sc Wr-own}).  
Other threads $\tida$ will have 
(2) their lists being split into ``old'' values for $\loc$ with $\lR$ flag and the new value for $\loc$ ({\sc Wr-other-1}), 
(3) properties (expressed as $\inter_\tid$) of suffixes of lists being preserved when the writing thread satisfies the same properties ({\sc Wr-other-2}) 
and 
(4) their lists consisting of $\lR$-accesses to $\loc$ followed by properties of the writer ({\sc Wr-other-3}).  
The last axiom concerns $\swapInstn$ only: as it can only read from store entries marked as $\lU$ it discards intervals satisfying $[\lR(\loc)]$. 

\begin{example}
We employ the axioms for showing one proof step for MP, namely one pair in the non-interference check 
of the rely $\cR_2$ of $\ctid{2}$ with respect to the guarantees $\cG_1$ of $\ctid{1}$:
%$\cR_2$ contains the assertion $\potassert{\ctid{2}}{[\cloc{y} \neq 1] \chop [\cloc{x}=1]}$ 
%and $\cG_1$ contains $\gc{\potassert{\ctid{1}}{[\cloc{x}=1]}}{\ctid{1}\mapsto \writeInst{\cloc{x}}{1}}$. The memory triple to be shown is 
$$\srahoare{\potassert{\ctid{2}}{[\cloc{y} \neq 1] \chop [\cloc{x}=1]} \land \potassert{\ctid{1}}{[\cloc{x}=1]}}{\ctid{1}\mapsto \writeInst{\cloc{x}}{1}}{\potassert{\ctid{2}}{[\cloc{y} \neq 1] \chop [\cloc{x}=1]}}$$
By taking $\inter_\tid$ to be $[\cloc{x}=1]$, this is an instance of {\sc Wr-other-2}. 
\end{example}

 In addition to the axioms above, we use
a {\em shift} rule for load instructions:
$$\inferrule*[left=Ld-shift]{\srahoare{\potassert{\tid}{\inter}}{\tid \mapsto \readInst{\reg}{\loc}}{\postf} \\ \reg \notin \fv(\inter)}
     {\srahoare{\potassert{\tid}{[ (\exp \wedge \eexp)\subst{\reg}{\loc}] \chop \inter}}{\tid \mapsto \readInst{\reg}{\loc}}{(\exp \land \potassert{\tid}{[\eexp] ; I}) \lor \postf}}$$
A load instruction reads from the first store in the lists, however, if the list satisfying $[ (\exp \wedge \eexp)\subst{\reg}{\loc}]$ in $[ (\exp \wedge \eexp)\subst{\reg}{\loc}] \chop \inter$ is empty, it reads from a list satisfying $\inter$. 
%\ori{lose is not a part of load. The reason here, I think, is that we never no if the [E] assertion is not satisfied by empty interval.}
%However, in interval assertions $[\eexp_1]\ldots [\eexp_n]$ a list prefix satisfying $[\eexp_1]$ might also be empty, so a load could take its value from the next list part satisfying $[\eexp_2]$ and so on.
The shift rule for $\readInstn$ puts this shifting to next stores  into a proof rule. 
Like the standard Hoare rule {\sc Subst-asgn}, {\sc Ld-shift} employs backward substitution. 
\begin{example}
We exemplify rule {\sc Ld-shift} on another proof step of example MP, one for local correctness of  $\ctid{2}$:
$$\small\srahoare{\potassert{\ctid{2}}{[\cloc{y} \neq 1]\chop[\cloc{x}=1]}}{\ctid{2}\mapsto \readInst{\creg{a}}{\cloc{y}}}{\creg{a} =1 \implies \potassert{\ctid{2}}{[\cloc{x}=1]}}$$
From axiom {\sc Stable-ld} we get
$\srahoare{\potassert{\ctid{2}}{ [\cloc{x}=1]}}{\ctid{2}\mapsto
  \readInst{\creg{a}}{\cloc{y}}}{ \potassert{\ctid{2}}{[\cloc{x}=1]}}$.   
We obtain
$\srahoare{\potassert{\ctid{2}}{[\cloc{y} \neq 1]\chop[\cloc{x}=1]}}{\ctid{2}\mapsto
  \readInst{\creg{a}}{\cloc{y}}}{\creg{a} \neq 1 \lor
  \potassert{\ctid{2}}{[\cloc{x}=1]}}$ using the former as premise for
{\sc Ld-shift}.
\end{example}

In addition, we include the standard conjunction, disjunction and consequence rules of Hoare logic. % as detailed in \cref{def:m-logic-soundness}.
 For instrumented primitive commands we employ the following rule: 
$$\small\inferrule*[left=Instr]{\srahoare{\postf_0}{\tid \mapsto \pcmd}{\postf_1} \\ \srahoare{\postf_1}{\tid \mapsto \assignInst{\reg_1}{\exp_1}}{\postf_2}  \ldots \srahoare{\postf_{n-1}}{\tid \mapsto \assignInst{\reg_n}{\exp_n}}{\postf_n}}
{\srahoare{\postf_0  }{\tid \mapsto \tup{\pcmd,\assignInst{\tup{\reg_1 \til \reg_n}}{\tup{\exp_1 \til \exp_n}}}}{\postf_n}}$$

%Logical rules:
%\begin{mathpar} 
%\inferrule[consequence]{
%\pref \implies \pref' \\ \postf' \implies \postf \\\\
%\srahoare{\pref'}{\tid \mapsto \alpha}{\postf'}}
%{\srahoare{\pref}{\tid \mapsto \alpha}{\postf}}
%\hfill
%\inferrule[conjunction]{
%\srahoare{\pref_1}{\tid \mapsto \alpha}{\postf_1} \\\\
%\srahoare{\pref_2}{\tid \mapsto \alpha}{\postf_2}}
%{\srahoare{\pref_1\land\pref_2}{\tid \mapsto \alpha}{\postf_1\land\postf_2}}
%\hfill
%\inferrule[disjunction]{
%\srahoare{\pref_1}{\tid \mapsto \alpha}{\postf_1} \\\\
%\srahoare{\pref_2}{\tid \mapsto \alpha}{\postf_2}}
%{\srahoare{\pref_1\lor\pref_2}{\tid \mapsto \alpha}{\postf_1\lor\postf_2}}
%\end{mathpar} 

%\subsection{Theorem}

Finally, it can be shown that all triples derivable from axioms and rules are valid memory triples.  

\begin{lemma}
If a \lora memory triple is derivable, $\vdash_\lora \srahoare{\pref}{\tid \mapsto \alpha}{\postf}$, 
then $\SRA \vDash \mhoare{\set{\tup{\regstore,\DD} \st \tup{\regstore,\DD} \models \pref}}{\tid \mapsto \alpha}{\set {\tup{\regstore,\DD} \st \tup{\regstore,\DD} \models \postf}}$.
\end{lemma}

%An immediate corollary is that if $\srahoare{\pref \land \pref'}{\tid \mapsto \alpha}{\postf}$ is derivable,
%then $\SRA \vDash \mhoare{\set{\astate \in \allstates_\M \st \astate \models \pref} \cap  \set{\astate \in \allstates_\M \st \astate \models \pref'}}{\tid \mapsto \alpha}{\set {\astate \in \allstates_\M \st \astate \models \postf}}$,
%so we can use the proof system to check non-interference conditions as well (see \cref{def:non-interfering}).
%
%To prove the lemma, we use \cref{def:m-logic-soundness} and \cref{def:m-soundness}...
%TODO

%%% Local Variables:
%%% mode: latex
%%% TeX-master: "main"
%%% End:

%% file: examples.tex
\newcommand{\ltupbracket}{\langle\!\langle}
\newcommand{\rtupbracket}{\rangle\!\rangle}

\newcommand{\tuptwo}[3]{#1 \in \ltupbracket #2; #3\rtupbracket}
\newcommand{\tupthree}[4]{#1 \in \ltupbracket#2; #3; #4\rtupbracket}

\begin{figure}[t]
  \centering
  \noindent
  \begin{minipage}[b]{\columnwidth}
    \centering
      \noindent 
  \begin{tabular}[b]{@{}c@{}} 
    $\assert{\potassert{\ctid{0}}{[\cloc{x} \neq 2]}}$ \qquad\qquad {} \\ 
    $\begin{array}{@{}l@{\ }||@{\ }l}
       \begin{array}[t]{l}
         \textbf{Thread } \ctid{1} 
         \\
         
         \assert{\potassert{\ctid{1}}{I^\cloc{x}_0}} \\
         
         1: \writeInst{\cloc{x}}{1}; \\ 
         
         \assert{\potassert{\ctid{1}}{I^\cloc{x}_1}} \\

         2: \writeInst{\cloc{x}}{2} \\

         \assert{\true} \\

       \end{array}
     & 
         \begin{array}[t]{l}
           \textbf{Thread } \ctid{2} 
           \\
           
           \assert{\potassert{\ctid{2}}{\inarrT{I_{012}^\cloc{x}}}} \\

           3: \readInst{\creg{a}}{\cloc{x}}  ;
           \\
           
           \assert{\cloc{a} = 2 \implies  \potassert{\ctid{2}}{I^\cloc{x}_2}} \\
           
           4: \readInst{\creg{b}}{\cloc{x}}   \\
           
           \assert{\cloc{a} = 2 \implies \cloc{b} = 2} \\
           
         \end{array}
     \end{array}$ \\
    $\assert{\creg{a}=2 \implies \creg{b} \neq 1} $\qquad\qquad {} 
  \end{tabular}
%\vspace{-10pt}
   \caption{RRC for two threads (\aka CoRR0)}
   \label{fig:rrc2}
%\vspace{-10pt}
 \end{minipage}
 
 \bigskip
 \hfill 
 \begin{minipage}[b]{\columnwidth}
   \scalebox{0.95}{
   \begin{tabular}[b]{@{}c@{}} 
        $\assert{\potassert{\ctid{0}}{I^\cloc{x}_0}}$  {} \\ 
       \hspace*{-5pt}$\begin{array}{l@{\ }||@{\ }l||@{\ }l||@{\ }l}
          \begin{array}[t]{l}
            \textbf{Thread } \ctid{1} 
            \\
            
            \assert{\bigwedge_{i\in\{1,3,4\}}
 \potassert{\ctid{i}}{I^{\cloc{x}}_{02}}  \\ {} \wedge \creg{c} \neq 1} \\
            
            1: \writeInst{\cloc{x}}{1} \\

            \assert{\true} \\

          \end{array}
     & 
         \begin{array}[t]{l}
           \textbf{Thread } \ctid{2}
           \\
           \assert{\bigwedge_{i\in\{2,3,4\}} \potassert{\ctid{i}}{I^{\cloc{x}}_{01}}  \\ {} \wedge \creg{a} \neq 2} \\

         2: \writeInst{\cloc{x}}{2} \\
           
         \assert{\true} \\
           
         \end{array}
                        & 
        \begin{array}[t]{l}
          \textbf{Thread } \ctid{3} 
          \\
          
           \assert{\potassert{\ctid{3}}{(I^{\cloc{x}}_{012} \lor I^{\cloc{x}}_{021})}}
           \\
          
          3: \readInst{\creg{a}}{\cloc{x}}  ;
          \\
          
          \assert{\cloc{a} = 2 \implies 
          \potassert{\ctid{3}}{I^{\cloc{x}}_{21}}} \\
          
          4: \readInst{\creg{b}}{\cloc{x}}   \\
          
          \assert{\tup{\cloc{a}, \cloc{b}} = \tup{2, 1} \implies \ \ \\
          \hfill \potassert{\ctid{3}}{I^\cloc{x}_1}} \\

        \end{array}
          & 
          \begin{array}[t]{l}
           \textbf{Thread } \ctid{4}
           \\
           
           \assert{\potassert{\ctid{4}}{(I^{\cloc{x}}_{012} \lor I^{\cloc{x}}_{021})}}
           % \lor {} 
           % \potassert{\ctid{4}}{}}
           \\
           
           5: \readInst{\creg{c}}{\cloc{x}}  ;
           \\
           
           \assert{\cloc{c} = 1 \implies 
                   \potassert{\ctid{4}}{I^{\cloc{x}}_{12}}} \\
           
           6: \readInst{\creg{d}}{\cloc{x}}   \\
           
            \assert{\tup{\cloc{c}, \cloc{d}} = \tup{1, 2} \implies \ \ \\
            \hfill \potassert{\ctid{4}}{I^\cloc{x}_2}} \\
           
         \end{array}
     \end{array}$ \\
    $\assert{\tup{\creg{a},\creg{b}}=\tup{2,1} \implies\tup{\cloc{c},\cloc{d}} \neq \tup{1,2}} $ {} 
   \end{tabular}}
% \vspace{-10pt}
   \caption{RRC for four threads (\aka CoRR2)} 
   \label{fig:rrc4}
%\vspace{-10pt}
 \end{minipage} 

\end{figure}

\section{Examples}

We discuss examples verified in \lora.  Additional examples can be
found in the appendix.

\paragraph{\textbf{Coherence.}} 
We provide two coherence examples in
\cref{fig:rrc2,fig:rrc4}, using the notation
$I^\loc_{\val_1 \val_2 \dots \val_n} = [x=\val_1] \chop [\loc=\val_2] \chop \ldots \chop [\loc=\val_n]$.
\cref{fig:rrc2} enforces an ordering on writes to the shared location
$\cloc{x}$ on thread $\ctid{1}$. The postcondition guarantees that 
after reading the second write, thread $\ctid{2}$ cannot read from the first.
  \cref{fig:rrc4} is similar, but the
writes to $\cloc{x}$ occur on two different threads. The postcondition
of the program guarantees that the two different threads agree on the
order of the writes. In particular if one reading thread (here
$\ctid{3}$) sees the value $2$ then $1$, it is impossible for the
other reading thread (here $\ctid{4}$) to see $1$ then $2$.

Potential assertions provide a compact and intuitive
mechanism for reasoning, e.g., in \cref{fig:rrc2}, the precondition of
line 3 precisely expresses the order of values available to thread
$\ctid{2}$. This presents an improvement over view-based
assertions~\cite{DBLP:conf/ecoop/DalvandiDDW19}, which required a
separate set of assertions to encode write order.

% - We can describe the ordering. In the view-based semantics, we can only talk about view maximal. 
% - View-based needs special assertions to talk about the modification.

%% file: peterson.tex
%!TEX root = main.tex

\begin{figure}[t]
 \centering
  \begin{tabular}[b]{@{}c@{}} 
  % $\qquad \qquad \qquad \qquad\qquad \qquad\qquad \qquad\qquad \qquad\qquad  \qquad  \assert{\neg \creg{a_1}1 \wedge \neg \creg{a}_2} $ \\
  $\begin{array}{@{}l@{\ \ }@{\ \ }l  % ||@{\ \ }l
     }
    \begin{array}[t]{@{}l}
     \textbf{Thread } \ctid{1} 
     \\
      \assert{\neg \creg{a}_1 \wedge \neg \creg{a}_2 \wedge \creg{mx_1} = 0} \\
     \texttt{while } \neg \cloc{stop} \texttt{ do } 
     \  \assert{\neg \creg{a}_1 \wedge  (\neg  \creg{a}_2 \lor \potassert{\ctid{1}}{[ \lR(\cloc{turn}) ]\chop [ \cloc{flag}_2]} )} \\
     %\wedge {}    P_2 = [\cloc{turn}\mapsto_R ? ][flag_1 \mapsto 1]} \\ 
         1: \quad \writeInst{\cloc{flag}_1}{\btrue}; %\\ 
     %\qquad 
     \assert{\neg  \creg{a}_1 \wedge \potassert{\ctid{1}}{[ \cloc{flag}_1]} \wedge {} (\neg  \creg{a}_2 \lor \potassert{\ctid{1}}{[ \lR(\cloc{turn}) ]\chop [\cloc{flag}_2]})}  %\wedge {}    P_2 = [\cloc{turn}\mapsto_R ? ][flag_1 \mapsto 1]}
     \\
     2: \quad \langle \swapInst{\cloc{turn}}{2};  \assignInst{\creg{a}_1}{\btrue}% ; \widehat{\cloc{turn}} := 2
      \rangle; \\ 
     3: \quad \texttt{do} 
     \  \assert{ \creg{a}_1 \wedge (\neg  \creg{a}_2 
     \lor \potassert{\ctid{1}}{[ \cloc{flag}_2 \wedge \cloc{turn} \neq 1 ]} 
 \lor P)}  
\\
     4: \qquad \readInst{\creg{fl}_1}{\cloc{flag}_2}; %  \\
     % \qquad\quad
      \assert{ \creg{a}_1 \wedge (\neg  \creg{a}_2 
     \lor ( \creg{fl}_1 \wedge \potassert{\ctid{1}}{[ \cloc{flag}_2 \wedge  \cloc{turn} \neq 1]}) \lor P)} % \\ 
\\
     5: \qquad \readInst{ \creg{tu}_1}{\cloc{turn}}; % \\
     % \qquad\quad
      \assert{ \creg{a}_1 \wedge (\neg  \creg{a}_2 
      \lor ( \creg{fl}_1 \wedge  \creg{tu}_1 \!\neq\! 1 \wedge \potassert{\ctid{1}}{[\cloc{flag}_2 \wedge \cloc{turn} \!\neq\! 1]}) \lor P)}
     %  \\ 
     % \qquad \qquad  {}
     % \lor (\potassert{\ctid{1}}{[\lR(\cloc{turn})]\chop[\cloc{flag}_2=1 \land \cloc{turn} =1]}))} 
     \\
     6: \quad \texttt{until} \ \neg \creg{fl}_1 \vee  (\creg{tu}_1 = 1) ; % \\
      % \qquad
      \assert{ \creg{a}_1 \wedge (\neg  \creg{a}_2 \lor P)}\\
     7: \quad \writeInst{\cloc{cs}}{\bot}; % \\
      % \qquad
      \assert{ \creg{a}_1 \wedge (\neg  \creg{a}_2 \lor P)} \\
     8: \quad \writeInst{\cloc{cs}}{0}; % \\ 
      % \qquad
      \assert{\potassert{\ctid{1}}{[ \cloc{cs}=0]} \wedge  \creg{a}_1 \wedge (\neg  \creg{a}_2 \lor P)} \\
     9: \quad \readInst{ \creg{mx_1}}{\cloc{cs}}; % \\ 
      % \qquad
      \assert{ \creg{mx_1} = 0 \wedge   \creg{a}_1 \wedge (\neg  \creg{a}_2 \lor P)} \\
     10: \;\langle \writeInst{\cloc{flag}_1}{0};  \assignInst{\creg{a}_1}{\bfalse} \rangle  \\
     \assert{ \creg{mx_1} = 0 % \wedge \neg  \creg{a}_1 \wedge  (\neg  \creg{a}_2 \lor \potassert{\ctid{1}}{[ \lR(\cloc{turn})]\chop [ \cloc{flag}_2=1]} )
      } \\
     % \texttt{od} 
   \end{array}
   \end{array}$
  \end{tabular} 
%  \vspace{-10pt}
 % \hspace*{8.9cm} $\assert{\creg{mx_1} = 0 \wedge \creg{mx_2} = 0 }$
 
 \caption{Peterson's algorithm, where $P = \potassert{\ctid{1}}{[\lR(\cloc{turn})]\chop[\cloc{flag}_2 \land \cloc{turn} =1]}$. Thread $\ctid{2}$ is symmetric and we assume a stopper thread
   $\ctid{3}$ that sets $\cloc{stop}$ to $\btrue$.}
 \label{fig:peterson}
%  \vspace{-10pt}
 
 \end{figure}

\paragraph{\textbf{Peterson's algorithm.}}
\Cref{fig:peterson} shows Peterson's
algorithm for implementing mutual exclusion for two threads~\cite{DBLP:journals/ipl/Peterson81} together with \lora
assertions. We depict only the code of thread $\ctid{1}$. Thread
$\ctid{2}$ is symmetric. A third thread $\ctid{3}$ is assumed
 stopping the other two threads at an arbitrary point
in time.  We use $\texttt{do } \cmd \texttt{ until } \exp$ as a
shorthand for $\cmd \sep \texttt{while } \exp \texttt{ do } \cmd$.
For correctness under \SRA, all accesses to the shared variable $\cloc{turn}$ are via a
$\swapInstn$, which ensures that $\cloc{turn}$ behaves like an $\SC$ variable.

Correctness is encoded via registers
$\creg{mx_1}$ and $ \creg{mx_2}$ into which the contents of shared
variable $\cloc{cs}$ is loaded. Mutual exclusion should guarantee both
registers to be 0. Thus neither threads should ever be able to read
$\cloc{cs}$ to be $\bot$ (as stored in line 7). 
 The proof (like
the associated $\SC$ proof in~\cite{DBLP:series/txcs/AptBO09}) introduces
auxiliary variables $\cloc{a_1}$ and $\cloc{a_2}$. Variable
$\cloc{a}_i$ is initially $\bfalse$, set to $\btrue$ when a thread
$\ctid{i}$ has performed its swap, and back to $\bfalse$ when
$\ctid{i}$ completes.
 % auxiliary variables $a_1, a_2$ (are the threads ``after'' setting $\cloc{turn}$?) and $\widehat{\cloc{turn}}$ 

Once again potentials provide convenient mechanisms for reasoning
about the interactions between the two threads. % , precisely recording
% the values that the other thread is allowed to return.
For example, the assertion
$\potassert{\ctid{1}}{[ \lR(\cloc{turn}) ]\chop [\cloc{flag}_2]}$ in
the precondition of line 2 encapsulates the idea that an RMW on
$\cloc{turn}$ (via $\swapInst{\cloc{turn}}{2}$) must read from a
  state in which $\cloc{flag}_2$ holds, allowing us to establish
  $\potassert{\ctid{1}}{[ \cloc{flag}_2 ]}$ as a postcondition (using
  the axiom {\sc Swap-skip}). We obtain disjunct
  $\potassert{\ctid{1}}{[ \cloc{flag}_2 \wedge \cloc{turn} \neq 1 ]}$
  after additionally applying {\sc Wr-own}.

%% file: related.tex
\section{Discussion, Related and Future Work} 

Previous RG-like logics provided ad-hoc solutions for other concrete
memory models such as x86-TSO and
C/C++11~\cite{DBLP:conf/icalp/LahavV15,DBLP:journals/jar/DalvandiDDW22,DBLP:conf/esop/BilaDLRW22,DBLP:journals/pacmpl/RaadLV20,DBLP:conf/fm/WrightBD21,DBLP:conf/ecoop/DalvandiDDW19,DBLP:conf/vstte/Ridge10}.
These approaches established soundness of the proposed logic with an ad-hoc proof that couples 
together memory and thread transitions.
We believe that these logics can be formulated in our proposed general RG framework
(which will require extensions to other memory operations such as fences).

Moreover, Owicki-Gries logics for different fragments of the C11 memory
model~\cite{DBLP:journals/jar/DalvandiDDW22,DBLP:conf/fm/WrightBD21,DBLP:conf/ecoop/DalvandiDDW19}
used specialized assertions over the underlying view-based
semantics. These include {\em conditional-view assertion} (enabling
reasoning about MP), and {\em value-order} (enabling reasoning about
coherence). Both types of assertions are special cases of the
potential-based assertions of \lora.

Ridge~\cite{DBLP:conf/vstte/Ridge10} presents an RG reasoning
technique tailored to x86-TSO, treating the write buffers in TSO
architectures as threads whose steps have to preserve relies. This is
similar to our notion of stability of relies under internal memory
transitions. Ridge moreover allows to have memory-model specific
assertions (\eg on the contents of write buffers).

The OGRA logic~\cite{DBLP:conf/icalp/LahavV15} for Release-Acquire
(which is slightly weaker form of causal consistency compared to SRA studied in this paper)
takes a different approach, which cannot be directly handled in our framework.
It employs simple SC-like assertions at the price of having a non-standard
non-interference condition which require a stronger form of stability.

Coughlin \etal~\cite{DBLP:conf/fm/CoughlinWS21,Coughlin-nonMCA}
provide an RG reasoning technique for weak memory models with a
semantics defined in terms of {\em reordering relations} (on
instructions). They study both multicopy and non-multicopy atomic
architectures, but in all models, the rely-guarantee assertions are
interpreted over SC.

% provide no equivalence proofs between reordering and standard
% semantics of weak models.

Schellhorn \etal~\cite{DBLP:journals/amai/SchellhornTEPR14} develop a
framework that extends ITL with a compositional interleaving operator,
enabling proof decomposition using RG rules. Each interval represents
a sequence of states, strictly alternating between program and
environment actions (which may be a skip action). This work is
radically different from ours since (1) their states are interpreted
using a standard SC semantics, and (2) their intervals represent an
{\em entire execution} of a command as well the interference from the
environment while executing that command.

Under SC, rely-guarantee was combined with separation logic~\cite{DBLP:phd/ethos/Vafeiadis08,DBLP:conf/concur/VafeiadisP07},
which allows the powerful synergy of reasoning using stable invariants (as in rely-guarantee)
and ownership transfer (as in concurrent separation logic).
It is interesting to study a combination of our RG framework with concurrent separation logics
for weak memory models, such as~\cite{DBLP:conf/oopsla/VafeiadisN13,DBLP:conf/esop/SvendsenPDLV18}.

Other works have studied the decidability of
verification for causal consistency models. In work preceding the potential-based SRA
model~\cite{DBLP:journals/toplas/LahavB22}, Abdulla \etal~\cite{DBLP:conf/pldi/AbdullaAAK19} show that verification under RA
is undecidable. In other work, Abdulla \etal~\cite{DBLP:conf/netys/AbdullaABKS22} show that the reachability
problem under TSO remains decidable for systems with dynamic thread
creation. Investigating this question under SRA is an interesting
topic for future work.

Finally, the spirit of our generic approach is similar to
Iris~\cite{DBLP:journals/jfp/JungKJBBD18},
Views~\cite{DBLP:conf/popl/Dinsdale-YoungBGPY13}, Ogre and Pythia~\cite{DBLP:conf/popl/AlglaveC17}, the work of Ponce de Le{\'{o}}n \etal~\cite{DBLP:conf/fmcad/LeonFHM18}, and recent axiomatic
characterizations of weak memory
reasoning~\cite{DBLP:journals/tocl/DohertyDDW22}, which all aim to provide a
{\em generic} framework that can be instantiated to underlying semantics. 

%\ori{this paragraph can go if there is no space, and the change the section title}
In the future we are interested in automating the reasoning in \lora, starting from automatically 
checking for validity of program derivations (using, \eg SMT solvers for specialised theories of sequences or strings~\cite{DBLP:conf/cade/ShengNRZDGPQBT22,DBLP:conf/cpp/KanLRS22}),
and, including, more ambitiously, synthesizing appropriate \lora invariants.

%%% Local Variables:
%%% mode: latex
%%% TeX-master: "main"
%%% End:

%% file: appendix.tex
%!TEX root = main.tex

\section{Auxiliary Variables and Rule of Consequence}

Here, we provide the necessary definitions and the rule for auxiliary variables. 

\begin{definition}
For a set $\regset \subseteq \Reg$, 
two register stores $\regstore$ and $\regstore'$ are \emph{$\regset$-equivalent}, 
denoted by $\regstore =_\regset \regstore'$, 
if $\regstore(\reg) = \regstore'(\reg)$ for every $\reg \in \regset$.
This equivalence is lifted to pairs of register stores and memory states 
by defining $\tup{\regstore,\memstate}=_\regset \tup{\regstore',\memstate'}$ if 
$\regstore =_\regset \regstore'$ and $\memstate = \memstate'$. 
\end{definition}

\begin{definition}
The removal of a set $\regset \subseteq \Reg$ from a multi-assignment 
$\assignInst{\vec{\reg}}{\vec{\exp}}$,
denoted by $\rem{\assignInst{\vec{\reg}}{\vec{\exp}}}{\regset}$, 
is the multi-assignment 
$\assignInst{\tup{\reg_{j_1}\til\reg_{j_m}}}{\tup{\exp_{j_1} \til \exp_{j_n}}}$
where $j_1 < \ldots < j_m$ is an enumeration of $\set{1 \leq j \leq n \st \reg_j\in \regset}$.
This removal is only defined if no register in $\regset$ occurs in 
$\tup{\exp_{j_1} \til \exp_{j_n}}$.
This notation is extended to commands in the expected way by applying 
removal inside the second component of instrumented primitive commands.
%If this second component becomes empty, we replace the instrumented command
%by a standard command.
It is only defined if all registers in $\regset$ only occur
as instrumentation, \ie in the second component of instrumented primitive commands.
\end{definition}

%For example, 
%$$\rem{\ite{\creg{1}}{ \readInst{\creg{4}}{\loc}}{\tup{\assignInst{\creg{2}}{\creg{1}} ,\assignInst{\creg{3}}{\creg{1}}}}}{\set{\creg{3}}}
%=\ite{\creg{1}}{ \readInst{\creg{4}}{\loc}}{\assignInst{\creg{2}}{\creg{1}}}$$
%while
%$\rem{\ite{\creg{1}}{ \readInst{\creg{4}}{\loc}}{\tup{\assignInst{\creg{2}}{\creg{1}} ,\assignInst{\creg{3}}{\creg{1}}}}}{\set{\creg{4}}}$
%is undefined.

\begin{definition}
A set $\prop \suq \allstates$ of states is \emph{independent of a set $\regset \subseteq \Reg$},
denoted by $\ind{\prop}{\regset}$,
if for every two states $\astate$ and $\astate' \in \allstates_\M$ with 
$\astate =_{\Reg \setminus \regset} \astate'$,
we have $\astate \in \prop \iff \astate' \in \prop$.
This notion is lifted to guarded commands by defining
$\ind{\gc{\prop}{\tid \mapsto c}}{\regset} \defiff \ind{\prop}{\regset}$,
and to sets of sets of states (or of  guarded commands)
by requiring that every element of the set is independent of $\regset$.
\end{definition}

\noindent The rule of auxiliary variables then takes the following form: 

\[ \inferrule[Aux]{
\set{\tid \mapsto \cmd'} \satm (\pre', \cR', \cG, \post) \\\\
\cmd = \rem{\cmd'}{\regset} \\\\
\ind{\pre}{\regset} \\ \ind{\cR}{\regset} \\ \ind{\cG}{\regset} \\ \ind{\post}{\regset} \\\\
\forall \astate \in \pre \ldotp \exists \astate' \in \pre' \ldotp \astate =_{\Reg \setminus \regset} \astate' \\\\
\inarr{ \forall \astate_1,\astate_2 \ldotp 
(\forall \statesetR\in\cR  \ldotp \astate_1\in \statesetR \implies \astate_2 \in \statesetR)
\implies \\
\qquad\qquad
\forall \astate_1' \ldotp \astate_1 =_{\Reg \setminus \regset} \astate_1' \implies 
\exists \sigma_2'\ldotp \astate_2 =_{\Reg \setminus \regset} \astate_2' \land
(\forall \statesetR'\in\cR'  \ldotp \astate_1'\in \statesetR' \implies \astate_2' \in \statesetR')}
}{\set{\tid \mapsto \cmd} \satm (\pre,\cR,\cG,\post)} 
\]

\noindent The proof rules also contain a rule of consequence: 

\[
\inferrule[consq]{
\set{\tid \mapsto \cmd} \satm (\pre',\cR', \cG', \post') \\
\pre \subseteq \pre' \\
\forall \statesetR'\in\cR', \astate \in \statesetR', \astate'\in \allstates \ldotp (\forall \statesetR \in \cR \ldotp \sigma \in \statesetR \implies \sigma' \in \statesetR) \implies  \astate'\in \statesetR'\\
%\cR' \subseteq \cR \\
\post' \subseteq \post\\\forall \gc{\prop'}{\tida \mapsto \ipcmd} \in \cG' \ldotp 
\exists \prop   \ldotp \prop' \subseteq \prop \land \gc{\prop}{\tida \mapsto \ipcmd} \in \cG}
{\set{\tid \mapsto \cmd} \satm (\pre,\cR, \cG, \post)}
\]

\section{Soundness of RG Proof Rules}

\RGsoundness*

\noindent 
\begin{proof} 
By induction on the structure of rules. Next, we consider the rules in Fig.~\ref{fig:proof-rules} plus the rule for Fork/Join and consequence. The rules for parallel composition and auxiliary variables are treated below. 
%\hw{look at $\epsl$ steps and decide whether to put them into guarantee.} \hw{or, better, assume something about G, like the reflexivity we had before}
\begin{description}
   \item[Skip] Let $$\xi \in \Comp(\{ \tid \mapsto \skipc\}) \cap \assume(\pre,\{\pre\}).$$ We need to prove that $\xi \in \commit(\emptyset,\pre)$. \\
   The computation $\xi$ has $\memstep$ and $\envstep$ transitions only. By $(\ref{int-rely})$ and $\xi \in \assume(\pre,\{\pre\})$, both type of transitions preserve $\pre$. Hence, in case $\xi$ is finite, $\tup{\regstore_{\last(\comp)},\memstate_{\last(\comp)}}\in\pre$
 and furthermore 
$\cmdmap_{\last(\comp)}(\tid') = \skipc$ for every $\tid'\in\dom{\cmdmap_{\last(\comp)}}$ as $\dom{\cmdmap_{\last(\comp)}} = \{\tid\}$. $\xi$ also admits $\cG$ as $\cG = \emptyset$. 
\item[Instrumented primitive commands] Let
  $$
  \comp \in \Comp(\{ \tid \mapsto \ipcmd \}) \cap \assume(\pre,\{\pre,\post\}).
  $$
  We need to prove that $\xi \in \commit(\gc{\pre}{\tid \mapsto \ipcmd}, \post)$. \\
   By the operational semantics, $\xi$ has to take the form 
   \[ \tup{\{\tid \mapsto \ipcmd \},\regstore_0,\memstate_0} \trans {a_1} \tup{\{\tid \mapsto \ipcmd \} ,\regstore_1,\memstate_1} \trans {a_2} \ldots \trans {a_j} \tup{\{\tid \mapsto \skipc \},\regstore_j,\memstate_j}  \ldots \] 
    with $a_j = \compstep$, $a_i \in \{ \memstep, \envstep \}$ for all $i, 1 \leq i \leq j-1 \vee i \geq j+1$, and for all $m \geq j+1$ the command map is $\{\tid \mapsto \skipc \}$. By $\xi \in \assume(\pre,\{\pre,\post\})$ and (\ref{int-rely}), $\tup{\regstore_0, \memstate_0}$ and also $\tup{\regstore_{j-1}, \memstate_{j-1}} \in \pre$. By $\M \vDash \mhoare{\pre}{\tid \mapsto \ipcmd}{\post}$,  
    %$\{pre\} \tid \mapsto \pcmd \{\post\}$, 
    $\tup{\regstore_{j}, \memstate_{j}} \in \post$. By $\xi \in \assume(\pre,\{\pre,\post\})$, $\tup{\regstore_m,\memstate_m} \in \post$ for all $m \geq j$. 
    Hence, if $\xi$ is finite, $\tup{\regstore_{\last(\comp)},\memstate_{\last(\comp)}}\in\post$ and 
$\cmdmap_{\last(\comp)}(\tid) = \skipc$ for every $\tid\in\dom{\cmdmap_{\last(\comp)}}$. Finally, $\xi$ admits $\cG$ as $a_j$ is the only component step executing $\pcmd$ in a state satisfying $\pre$. 
   %\item[Instrumented primitive commands, read-modify-write instructions] Same as primitive commands. \hw{change rule in figure? changed, now adapt proof here} 
   \item[Sequential composition] Let $\xi \in \Comp(\set{\tid \mapsto\cmd_1 \seq \cmd_2}) \cap \assume(\pre,\cR_1 \cup \cR_2)$. We need to prove that $\xi \in \commit(\cG_1 \cup \cG_2,\post)$.  
   The computation $\comp = \tup{\cmdmap_0,\regstore_0,\memstate_0} \trans {a_1} \tup{\cmdmap_1,\regstore_1,\memstate_1} \trans {a_2} \ldots$ takes one of two forms: \\
   (1) Either (the $\cmd_1$ part is infinite) 
    $\cmdmap_i = \{\tid \mapsto \cmd_1^i \seq \cmd_2\}$, $i \geq 0$,  and $\comp' = \tup{\{\tid \mapsto \cmd_1\},\regstore_0,\memstate_0} \trans {a_1} \tup{\{\tid \mapsto \cmd_1^1,\regstore_1,\memstate_1} \trans {a_2} \ldots$ is in $\Comp(\{ \tid \mapsto \cmd_1\}) \cap \assume(\pre,\cR_1)$ (as $\cR_1 \subseteq \cR_1 \cup \cR_2$), \\ 
   or (2) ($\cmd_1$ part finite, $\cmd_2$ part finite or infinite) 
  there exists $k$ s.t.~$\cmdmap_k = \{ \tid \mapsto \skipc \seq \cmd_2\}$, $\cmdmap_{k+1} = \{ \tid \mapsto   \cmd_2\}$ and $a_{k+1} = \epsl$. \\ 
   First case: By $\{\tid \mapsto \cmd_1\} \satm (\pre, \cR_1,\cG_1,mid)$, $\comp' \in \commit(\cG_1,mid)$. Hence $\comp' \in \commit(\cG_1 \cup \cG_2,mid)$, and as $\comp'$ and $\comp$ are infinite, $\comp \in \commit(\cG_1 \cup \cG_2,\post)$ (no final state). \\
   Second case: We split $\comp$ into $\comp_1$ and $\comp_2$ by letting $\comp_1$ run until state $k$, $\tup{\{ \tid \mapsto \skipc \seq \cmd_2\},\regstore_k, \memstate_k}$, and have $\comp_2$ start in state $k+1$, $\tup{\{ \tid \mapsto \cmd_2\},\regstore_{k+1}, \memstate_{k+1}}$. From $\comp_1$ we construct computation $\comp_1'$ like for case (1). 
   Again, $\comp_1' \in \Comp(\{ \tid \mapsto \cmd_1\}) \cap \assume(\pre,\cR_1)$. As $\comp_1'$ is finite and by $\{ \tid \mapsto\cmd_1 \} \satm (\pre,\cR_1,\cG_1,mid)$, $\comp_1' \in \commit(\cG_1, mid)$, hence $\tup{\regstore_k,\memstate_k} \in mid$ and so is $\tup{\regstore_{k+1},\memstate_{k+1}}$ (as the $\epsl$-step of $\skipc$ does not change registers nor memory). Furthermore, all component steps in $\comp_1'$ satisfy $\cG_1$ (and thus $\cG_1 \cup \cG_2$) and so do the same steps in $\comp$. 
   Then, $\comp_2 \in \Comp(\{ \tid \mapsto \cmd_2) \cap \assume(mid,\cR_2)$, hence by $\{ \tid \mapsto\cmd_2 \} \satm (mid, \cR_2,\cG_2,\post)$, $\comp_2 \in \commit(\cG_2,\post)$. Thus, if $\comp_2$ is finite,  $\tup{\regstore_{\last(\comp_2)},\memstate_{\last(\comp_2)}}\in\post$ and 
$\cmdmap_{\last(\comp)}(\tid) = \skipc$. Hence $\comp \in \commit(\cG_1 \cup \cG_2,\post)$.  
\item[If] Let $\comp \in \Comp(\set{\tid \mapsto \ite{\exp}{\cmd_1}{\cmd_2}}) \cap \assume(\pre,\cR_1 \cup \cR_2 \cup \set{\pre})$. Need to show $\comp \in \commit(\cG_1 \cup \cG_2,\post)$. Then there exists a $ k$ s.t.~$\cmdmap_k = \set{\tid \mapsto \ite{\exp}{\cmd_1}{\cmd_2}}$, $a_{k+1} = \epsl$  and for all $i \leq k$, step $a_i \neq \compstep$, hence $\tup{\regstore_k, \memstate_k} \in \pre$ (as $\pre$ is in the relies). Then two cases: (1) $\regstore_k \in \sem{\exp}$ or is not. The cases are dual and we just consider the first one. \\
In that case, $\cmdmap_{k+1} = \set{\tid \mapsto \cmd_1}$. Furthermore, $\tup{\regstore_{k+1},\memstate_{k+1}} \in \pre \cap \sem{\exp} \times \M.\lQ$ (the $\epsl$-step of if neither changing registers nor memory). All further component steps of $\comp$ now satisfy $\cG_1$ (by $\set{\tid \mapsto \cmd_1} \satm (\pre \cap (\sem{\exp} \times \M.\lQ), \cR_1, \cG_1, \post)$. Hence, if $\comp$ is finite,  $\tup{\regstore_{\last(\comp)},\memstate_{\last(\comp)}}\in\post$ and 
$\cmdmap_{\last(\comp)}(\tid) = \skipc$ and thus $\comp \in \commit(\cG_1,\post)$ and hence in $\commit(\cG_1 \cup \cG_2, \post)$. 
\item[While] Let $\comp \in \set{\tid \mapsto \while{\exp}{\cmd}} \cap \assume(\pre,\cR \cup \set{\pre,\post})$. We need to show $\comp \in \commit(\cG,\post)$. \\
We let $\cmdmap_w = \while{\exp}{\cmd}$, $\cmdmap_{if} = \ite{\exp}{(\cmd \sep \while{\exp}{\cmd})}{\skipc}$,  $\cmdmap_s = \set{\tid \mapsto \cmd \seq \while{\exp}{\cmd}}$ and $\cmdmap_{end} = \set{\tid \mapsto \skipc \seq \while{\exp}{\cmd}}$. By the operational semantics, the computation $\comp$ takes the following form or a prefix of it followed by non-component steps: 
\begin{align*}
 \tup{\cmdmap_w, \regstore_0,\memstate_0} \trans {} \ldots \trans {} \tup{\cmdmap_w, \regstore_{k_1},\memstate_{k_1}} \trans \compstep \tup{\cmdmap_{if},\regstore_{k_1+1},\memstate_{k_1+1}} \trans {} \ldots \trans {} \\ 
 \tup{\cmdmap_{if},\regstore_{m_1},\memstate_{m_1}} \trans \compstep \tup{\cmdmap_s,\regstore_{m_1+1}, \memstate_{m_1+1}} \trans {} \ldots \trans {} \tup{\cmdmap_{end},\regstore_{k_2 -1},\memstate_{k_2-1}} \\
 \trans {\compstep} \tup{\cmdmap_w,\regstore_{k_2},\memstate_{k_2}} \trans {} \ldots \trans {} \tup{\cmdmap_{if},\regstore_{m_2},\memstate_{m_2}} \trans \compstep \tup{\cmdmap_s,\regstore_{m_2+1}, \memstate_{m_2+1}}  \\  \trans {} \ldots \trans {} \tup{\cmdmap_w, \regstore_{k_r},\memstate_{k_r}} \trans \compstep \tup{\cmdmap_{if}, \regstore_{k_r+1},\memstate_{k_r+1}} \trans {} \ldots \trans {} % \tup{\cmdmap_{if},\regstore_{m_r},\memstate_{m_r}}
 \end{align*} 
 which goes on like this forever, or eventually reaches $\tup{\cmdmap_{if},\regstore_{m_l},\memstate_{m_l}}$ with $\tup{\regstore_{m_l},\memstate_{m_l}} \notin\sem{\exp} \times \M.\lQ$ and $\cmdmap_{m_l+1} = \set{\tid \mapsto \skipc}$. Afterwards, the computation either has $\envstep$ or $\memstep$ steps forever, or $\comp$ is final.  By the operational semantics and $\set{\tid \mapsto \cmd} \satm (\pre \cap (\sem{\exp} \times \M.\lQ),\cR,\cG,\pre)$ we get the following properties:
 \begin{align*}
   \forall i, 1 \leq i < l: & \quad \regstore_{m_i}(e) = 1 \\
     %\forall i, 1 \leq i < l: & \quad k_i < m_i \\
     \forall i, 1 \leq i < l: & \quad \tup{\regstore_{k_i}, \memstate_{k_i}} \in \pre \cap (\sem{\exp} \times \M.\lQ) \\ 
       & \quad \tup{\regstore_{m_l}, \memstate_{m_l}} \in \pre \setminus (\sem{\exp}\times \M.\lQ)
 \end{align*} 
 All components steps of $\comp$ are in $\cG$, either because they are steps of $\cmd$ or steps belonging to component steps labelled $\epsl$ which correspond to unfolding the while, evaluating the condition in if or moving from one part of a sequential composition to the next). By $\pre \setminus (\sem{\exp} \times \M.\lQ) \subseteq post$, $\tup{\regstore_{m_l},\memstate_{m_l}} \in \post$. Furthermore, by $\comp \in \assume(\pre,\cR \cup \set{\pre,\post})$, $\tup{\regstore_{\last(\comp)},\memstate_{\last(\comp)}}\in\post$ and 
$\cmdmap_{\last(\comp)}(\tid) = \skipc$. 
\item[Consequence] Let $\comp \in \Comp(\set{\tid \mapsto \cmd}) \cap \assume(\pre,\cR \cup \{\pre,\post\})$. We need to show $\comp \in \commit(\cG,\post)$. \\
As $\pre \subseteq \pre'$, $\comp \in \assume(\pre',\cR \cup \{\pre,\post\})$. \\
From $\comp$ admitting $\cR$, we get $\forall R \in \cR, \forall i \geq 0$ with $a_{i+1} = \envstep$: $\tup{\regstore_i,\memstate_i} \in R \implies \tup{\regstore_{i+1},\memstate_{i+1}} \in R$. Now let $R' \in \cR'$ and assume $\tup{\regstore_i,\memstate_i} \in R'$. By the condition on relies in the rule, we then get $\tup{\regstore_{i+1},\memstate_{i+1}} \in R'$. Hence $\comp \in \assume(\pre,\cR')$. \\
By $\set{\tid \mapsto \cmd} \satm (\pre',\cR', \cG', \post') $, we get $\comp \in \commit(\cG',\post')$. Thus for all $i\geq 0$ and $a_{i+1} = \compstep$, there exists $\gc{\prop'}{\tid \mapsto \ipcmd} \in  \cG'$ such that $\tup{\regstore_i,\memstate_i}\in \prop'$ and $\tup{\set{\tid\mapsto \ipcmd},\regstore_i, \memstate_i} \asteptidlab{\tid}{\lab_\epsl}{\cs{\M}}  \tup{\set{\tid\mapsto \skipc},\regstore_{i+1},\memstate_{i+1}}$
for some $\lab_\epsl\in\Lab \cup \set{\epsl}$ (and similar for Fork/Join steps).  
%$\tup{\regstore_i,\memstate_i} \trans {\tid, \lab} \tup{\regstore_{i+1},\memstate_{i+1}}$ for some label $\lab \in \labels(\ipcmd)$. 
By the condition on guarantees there is thus some $\gc{\prop}{\tid \mapsto \ipcmd} \in  \cG$ such that $\tup{\regstore_i,\memstate_i}\in \prop$   and $\tup{\set{\tid\mapsto \ipcmd},\regstore_i, \memstate_i} \asteptidlab{\tid}{\lab_\epsl}{\cs{\M}}  \tup{\set{\tid\mapsto \skipc},\regstore_{i+1},\memstate_{i+1}}$
for some $\lab_\epsl\in\Lab \cup \set{\epsl}$. Thus $\comp$ admits $\cG$. 
The rest follows from $\post' \subseteq \post$. 

\item[Fork/Join] Let $\comp \in \Comp(\set{\tid \mapsto \cmd_1 \spar{\tid_1}{}{\tid_2} \cmd_2}) \cap \assume(\pre, \cR \cup \set{\pre,\post})$. We need to show that $\comp \in \commit(\cG,\post)$. For the computation $\comp$, assume there is an index $k$ such that $\cmdmap_k = \set{\tid \mapsto \cmd_1 \spar{\tid_1}{}{\tid_2} \cmd_2}$, $\cmdmap_{k+1} = \set{\tid \mapsto \cmd_1 \spar{\tid_1}{}{\tid_2} \cmd_2} \uplus \set{\tid_1 \mapsto\cmd_1} \uplus \set{\tid_2\mapsto \cmd_2}$ and the step $a_{k+1}$ belongs to the component transition $\forklab{\tid_1}{\tid_2}$. Moreover, all steps prior to $k+1$ are non-$\compstep$ steps. Then $\tup{\regstore_k,\memstate_k} \in \pre$ (by $\pre$ being contained in the rely) and by $$\M \vDash \mhoare{\pre}{{\tid}\mapsto{\forklab{\tid_1}{\tid_2}}}{\pre'}$$ we get %$\{ \pre \} \tid \mapsto \forklab{\tid_1}{\tid_2} \{ pre' \} $ 
$\tup{\regstore_{k+1},\memstate_{k+1}} \in \pre'$. Now two cases: (1) either there is some $m > k$ such that $\cmdmap_m = \set{\tid \mapsto \cmd_1 \spar{\tid_1}{}{\tid_2}} \uplus \set{\tid_1 \mapsto\skipc} \uplus \set{\tid_2\mapsto \skipc}$ and the step $a_{m+1} = \compstep$, namely that belonging to $\joinlab{\tid_1}{\tid_2}$, and then $\cmdmap_{m+1} = \set{\tid \mapsto \skipc}$,  or (2) no such step exists and thus $\comp$ is infinite. \\
For case (1), $\tup{\regstore_m, \memstate_m} \in \post'$ and hence by $\M \vDash \mhoare{\post'}{{\tid}\mapsto{\joinlab{\tid_1}{\tid_2}}}{\post}$, we get $\tup{\regstore_{m+1}, \memstate_{m+1}} \in \post$ (and so are all further states because $\post$ is in the rely). 
In both cases, all component steps fulfill $\cG$ because they are either component steps of $\cmd_1$ or $\cmd_2$ or fork or join. Hence, $\comp \in \commit(\cG \cup \{\tid \mapsto \forkcmd{\tid_1}{\tid_2}, \tid \mapsto \joincmd{\tid_1}{\tid_2}\},\post)$. 
\end{description} 
\end{proof}

\paragraph{Soundness of parallel composition.} 
For parallel composition, we take a computation $\comp$ and make a projection onto the threads $\tid_1$ and $\tid_2$ (plus their forked threads). We let $T_i \subseteq \Tid$ contain $\tid_i$, $i=1,2$, plus threads (recursively) forked by it. We construct computations $\comp_1, \comp_2$ by replacing in $\comp = \tup{\cmdmap_0,\regstore_0,\memstate_0} \trans {a_1} \tup{\cmdmap_1,\regstore_1,\memstate_1} \trans {a_2} \ldots$ the command maps $\cmdmap_k$ by $\cmdmap_k|_{T_i}, i=1,2$ (making $\dom{\cmdmap_k} = T_i$).   All $\compstep$-steps of $T_2$ become $\envstep$-steps in $\comp_1$, and vice versa, other steps stay as they are. We get 
\[ \comp_i \in \Comp(\set{\tid_i \mapsto \cmd_i}), i=1,2 \] 

\noindent Then the following lemma is required for later showing soundness of the rule of parallel composition (to avoid circular reasoning). It is analogous to one of Xu et al.~\cite{DBLP:journals/fac/XuRH97}.   

\begin{lemma} \label{lem:project} 
 Let $\comp \in \Comp( \set{\tid_1 \mapsto\cmd_1} \uplus \set{\tid_2\mapsto \cmd_2}) \cap \assume(\pre,\cR_1 \cup \cR_2 \cup \{\pre,\post\})$ and $\comp_1, \comp_2$ be the projections onto threads $\tid_1, \tid_2$, respectively. Assume furthermore $\set{\tid_i \mapsto\cmd_i} \satm (\pre_i,\cR_i,\cG_i, \post_i)$,  $(\pre_1,\cR_1,\cG_1, \post_1)$  and $ (\pre_2,\cR_2,\cG_2,\post_2)$ to be non-interfering  and $
\pre \subseteq \pre_1 \cap \pre_2$. Then $\comp_1$ admits $\cG_1$ and $\comp_2$ admits $\cG_2$. 
\end{lemma}

\begin{proof}
 First of all, $\comp_i$ admits $\pre_i$ by $\pre \subseteq \pre_1 \cap \pre_2$ . Now proof by contradiction. Assume $j$ to be the smallest index such that for the transition $\trans {a_{j+1} = \compstep}$ in some $\comp_i$ there is no $\gc{p}{\tid_i \mapsto \ipcmd} \in \cG_i$ such that $\tup{\regstore_j,\memstate_j} \in p$ and $\tup{\regstore_j,\memstate_j} \trans {c} \tup{\regstore_{j+1},\memstate_{j+1}}$. Without loss of generality assume this to be thread $i=1$. We now consider the computation $\comp_1^{j+1}$  which is $\comp_1$ up to $\tup{\cmdmap_1^{j+1},\regstore_{j+1},\memstate_{j+1}}$. $\comp_1^{j+1} \in \Comp(\set{\tid_1 \mapsto \cmd_1})$. We now consider the $\envstep$-steps in $\comp_1^{j+1}$. These are either $\envstep$-steps in $\comp$ or $\compstep$-steps in $\comp_2$. In the first case, they preserve all $R \in \cR_1 \cup \cR_2$. In the latter case, there exists some $\gc{p'}{\tid_2 \mapsto \ipcmd'} \in \cG_2$ which justifies the step (if not, we would have found a smaller index). By non-interference, this step preserves all $R \in \cR_1$. Hence, $\comp_1^{j+1} \in \assume(\pre_1, \cR_1)$. By $\set{\tid_1 \mapsto\cmd_1} \satm (\pre_1,\cR_1,\cG_1, \post_1)$, we thus get $\comp_1^{j+1} \in \commit(\cG_1,\post_1)$ which gives the contradiction. 
 \end{proof} 
 
\noindent With this lemma at hand, we can show soundness of the rule for parallel composition. 
 \begin{proof}
 Let $\comp \in \Comp(\set{\tid_1 \mapsto\cmd_1} \uplus \set{\tid_2\mapsto \cmd_2}) \cap \assume(\pre,\cR_1 \cup \cR_2 \cup \{ \pre,\post\})$. By Lemma~\ref{lem:project} and its proof, we get $\comp_i \in \assume(\pre_i, \cR_i)$ for the two projections $\comp_1$ and $\comp_2$ of $\comp$. Hence, $\comp_i \in \commit(\cG_i, \post_i), i=1,2$. Furthermore, if $\comp$ is finite, so are $\comp_1$ and $\comp_2$. Then $\tup{\regstore_{\last(\comp_i)},\memstate_{\last(\comp)}}\in\post_i$ and 
$\cmdmap_{\last(\comp_i)}(\tid_i) = \skipc$, $i=1,2$. By $\post_1 \cap \post_2 \subseteq \post$, this final state also satisfies $\post$.  
 \end{proof} 
 
 \paragraph{Soundness of auxiliary variables.}
 We start by showing  a proposition about transitions to not be affected by auxiliary variables.  
 
 \begin{proposition} \label{prop:aux} 
   Let $\regset \subseteq \Reg$ be a set of auxiliary registers, $\regstore_1, \regstore_2$ register stores with $\regstore_1 =_{\Reg \setminus Z} \regstore_2$ and let  $\tup{\set{\tid \mapsto \rem{\ipcmd}{\regset}},\regstore_1,\memstate} \trans {\tid,\lab} \tup{\set{\tid \mapsto \skipc},\regstore_1',\memstate'}$. Then for all $\regstore_2'$ with  $\tup{\set{\tid \mapsto \ipcmd},\regstore_2,\memstate} \trans {\tid,\lab} \tup{\set{\tid \mapsto \skipc},\regstore_2',\memstate'}$ we have $\regstore_1' =_{\Reg \setminus Z} \regstore_2'$.  
   
   %\ipcmd ::=  &
% \tup{\pcmd,\assignInst{\vec{\reg}}{\vec{\exp}}}
 \end{proposition} 
 
 \begin{proof} Proof of the soundness of the rule of auxiliary variables.  \\
 \noindent  Let $\comp \in \Comp(\set{\tid \mapsto \cmd}) \cap \assume(\pre,\cR)$. We need to show that $\comp \in \commit(\cG, \post)$. Assume $\comp$ to be 
 \[ \tup{\{\tid \mapsto \cmd_0 \},\regstore_0,\memstate_0} \trans {a_1} \tup{\{\tid \mapsto \cmd_1 \} ,\regstore_1,\memstate_1} \trans {a_2} \ldots \trans {} \ldots  \] 
with $\cmd_0 = \cmd$.  Let $\regset \subseteq \Reg$ be a set of auxiliary registers and $\cmd'$ be an arbitrary command with $\rem{\cmd'}{\regset} = \cmd$ and $\set{\tid \mapsto \cmd'} \satm (\pre',\cR',\cG,\post)$. By the lifting of Proposition~\ref{prop:aux} to arbitrary commands, there exists a computation $\comp' \in \Comp(\set{\tid \mapsto \cmd'})$ with  
\[ \tup{\{\tid \mapsto \cmd_0' \},\regstore_0',\memstate_0} \trans {a_1} \tup{\{\tid \mapsto \cmd_1' \} ,\regstore_1',\memstate_1} \trans {a_2} \ldots \trans {} \ldots  \] 
such that $\cmd_0' = \cmd'$, $\rem{\cmd_i'}{\regset} = \cmd_i$ for all $i \geq 0$ and $\regstore_i =_{\Reg \setminus \regset} \regstore_i'$ for all $i \geq 0$.  By $\ind{\pre}{\regset}$, $\regstore_0' \in \pre$. We choose $\regstore_0'$ to be a register store satisfying $\pre'$ (which exists by condition $\forall \astate \in \pre \ldotp \exists \astate' \in \pre' \ldotp \astate =_{\Reg \setminus \regset} \astate'$). 
For the environment steps see below. 
We get $\comp' \in \assume(\pre',\cR)$ (using $\ind{\cR}{\regset}$).  \\
By $\comp'$ admitting $\cR$, we know that for all $R \in \cR$, $\tup{\regstore_i,\memstate_i} \trans {a_{i+1}} \tup{\regstore_{i+1},\memstate_{i+1}}$ with $a_{i+1} = \envstep$, the condition $\tup{\regstore_i,\memstate_i} \in R \implies \tup{\regstore_{i+1},\memstate_{i+1}} \in R$ holds. Now take arbitrary $R' \in \cR'$ and let $\tup{\regstore_i',\memstate_i} \in R'$. By the condition on relies,  we know the existence of $\tup{\regstore_2,\memstate_{i+1}}$ with 
$\tup{\regstore_{i+1},\memstate_{i+1}} =_{\Reg \setminus \regset} \tup{\regstore_2,\memstate_{i+1}}$ and $\tup{\regstore_2,\memstate_{i+1}} \in R'$. We take this to be $\tup{\regstore_{i+1}',\memstate_{i+1}}$. 
Hence $\comp' \in \assume(\pre',\cR')$. By $\set{\tid \mapsto \cmd'} \satm (\pre', \cR', \cG, \post)$ we get $\comp' \in \commit(\cG,\post)$. By $\ind{\cG}{\regset}$ and $\ind{\post}{\regset}$ we get $\comp \in \commit(\cG, \post)$. 
\end{proof}

%\endinput

\section{More Examples}

\subsection{Load-buffering (LB)} 

We present a proof outline for LB. 

\begin{center} 
  \begin{tabular}[b]{@{}c@{}} 
   $\assert{  \creg{a}=0 \land \creg{b}=0   } $ \\ 
  $\begin{array}{l@{\ \ }||@{\ \ }l}
    \begin{array}[t]{l}
     \textbf{Thread } 1 
     \\
     \assert{ \potassert{\ctid{2}}{[\cloc{y}= 0] \land \creg{b}=0  }} \\ %G_1(y \to 0)
     1: \readInst{\creg{a}}{\cloc{x}}; \\ 
     \assert{ \potassert{\ctid{2}}{ [\cloc{y} = 0] \land \creg{b}=0} }\\ % \land y \to 0
     2: \writeInst{\cloc{y}}{1}; \\
     \assert{ \creg{a}=0 \vee \creg{b}=0 } 
     \end{array}
    & 
    \begin{array}[t]{l}
     \textbf{Thread } 2 
     \\
     \assert{  \potassert{\ctid{1}} {[\cloc{x} = 0] \land \creg{a}=0 }} \\
     3: \readInst{\creg{b}}{\cloc{y}} ;
     \\
      \assert{ \potassert{\ctid{1}}{ [\cloc{x} = 0] \land \creg{a}=0 }} \\ 
     4: \writeInst{\cloc{x}}{1}; \\ 
     \assert{ \creg{a}=0 \vee \creg{b}=0 }
     \end{array}
     \end{array}$ \\
      $\assert{ \creg{a}=0 \vee \creg{b}=0 }  $
   \end{tabular} 
\end{center} 

\subsection{Two writers on two locations (2+2W)} 

We present a proof outline for 2+2W. 

 \begin{center} 
  \begin{tabular}[b]{@{}c@{}} 
   $\assert{\creg{a}=\creg{b}=\hat{\creg{c}} = 0} $ \\
  $\begin{array}{l@{\ \ }||@{\ \ }l}
    \begin{array}[t]{l}
     \textbf{Thread } \ctid{1} 
     \\
     \assert{\creg{a}=0} \\  
     1: \writeInst{\cloc{x}}{1}; \\ 
     \assert{  \true }     \\ 
     2: \writeInst{\cloc{y}}{2}; \\
      \assert{\potassert{\ctid{1}}{[ \cloc{y}=2]} \lor {} \\
              \left(\inarr{(\hat{\creg{c}}=1 \implies \potassert{\ctid{2}}{[ \cloc{x}=2]}) \land \\
     (\creg{b} \neq 0 \implies \creg{b}=2)}\right)} \\
     3: \readInst{\creg{a}}{\cloc{y}}; \\ 
     \assert{  \true }      \\
     \end{array}
    & 
    \begin{array}[t]{l}
     \textbf{Thread } \ctid{2} 
     \\
     \assert{\creg{b}=\hat{\creg{c}}=0} \\  
     4: \writeInst{\cloc{y}}{1}; \\ 
     \assert{  \true }     \\ 
     5: \tup{\writeInst{\cloc{x}}{2} , \assignInst{\hat{\creg{c}}}{1}} ; \\
     \assert{\potassert{\ctid{2}}{[ \cloc{x}\neq 0]} \land {} \\
              (\creg{a}=1 \implies \potassert{\ctid{2}}{[ \cloc{x}=2]}) \land \hat{\creg{c}} = 1} \\
     6: \readInst{\creg{b}}{\cloc{x}}; \\ 
     \assert{\creg{b} \neq 0 \land (\creg{a}=1 \implies \creg{b}=2)}
     \end{array}
     \end{array}$ \\
       $\assert{\creg{a}=1 \implies \creg{b}=2} $ 
   \end{tabular} 
 \end{center}
 
\subsection{Store-buffering with fences (SB)} 

SC-fences are modeled as RMWs to an otherwise unused location $\cloc{f}$.
The auxiliary variable $\hat{\creg{c}}$ is used to remember the order of the fence instructions.

 \begin{center} 
  \begin{tabular}[b]{@{}c@{}} 
   $\assert{\hat{\creg{c}}=0}$\\
  $\begin{array}{l@{\ \ }||@{\ \ }l}
    \begin{array}[t]{l}
     \textbf{Thread } \ctid{1} 
     \\
     \assert{\hat{\creg{c}}=0 \lor (\hat{\creg{c}}=2 \land \potassert{\ctid{1}}{[\lR(\cloc{f})]\chop[\cloc{y}=1]})} \\  
     1: \writeInst{\cloc{x}}{1}; \\ 
     \assert{\inarr{\potassert{\ctid{1}}[\cloc{x}=1] \land {}\\
      \inpar{\hat{\creg{c}}=0 \lor {} \\
             (\hat{\creg{c}}=2 \land \potassert{\ctid{1}}{[\lR(\cloc{f})]\chop[\cloc{y}=1]})}}}     \\ 
     2: \tup{\swapInst{\cloc{f}}{0} , \assignInst{\hat{\creg{c}}}{10 \hat{\creg{c}}+1}} ; \\
     \assert{\hat{\creg{c}}=1 \lor \hat{\creg{c}}=12 \lor {} \\
             (\hat{\creg{c}}=21 \land \potassert{\ctid{1}}{[\cloc{y}=1]})} \\
     3: \readInst{\creg{a}}{\cloc{y}}; \\ 
     \assert{\hat{\creg{c}}=1 \lor \hat{\creg{c}}=12 \lor (\hat{\creg{c}}=21 \land \creg{a}=1)}      \\
     \end{array}
    & 
    \begin{array}[t]{l}
     \textbf{Thread } \ctid{2} 
     \\
     \assert{\hat{\creg{c}}=0 \lor {} \\
             (\hat{\creg{c}}=1 \land \potassert{\ctid{2}}{[\lR(\cloc{f})]\chop[\cloc{x}=1]})} \\  
     4: \writeInst{\cloc{y}}{1}; \\ 
     \assert{\inarr{\potassert{\ctid{2}}[\cloc{y}=1] \land {}\\
     \inpar{\hat{\creg{c}}=0 \lor {}  \\
            (\hat{\creg{c}}=1 \land \potassert{\ctid{2}}{[\lR(\cloc{f})]\chop[\cloc{x}=1]})}}}    \\ 
     5: \tup{\swapInst{\cloc{f}}{0} , \assignInst{\hat{\creg{c}}}{10 \hat{\creg{c}}+2}} ; \\
     \assert{\hat{\creg{c}}=2 \lor \hat{\creg{c}}=21 \lor {} \\
              (\hat{\creg{c}}=12 \land \potassert{\ctid{2}}{[\cloc{x}=1]})} \\
     6: \readInst{\creg{b}}{\cloc{x}}; \\ 
     \assert{\hat{\creg{c}}=2 \lor \hat{\creg{c}}=21 \lor {} \\
             (\hat{\creg{c}}=12 \land \creg{b}=1)}
     \end{array}
     \end{array}$ \\
       $\assert{\creg{a}=1 \lor \creg{b}=1} $ 
   \end{tabular} 
 \end{center}

%%% Local Variables:
%%% mode: latex
%%% TeX-master: "main"
%%% End: